%% file: DecomposingNSChannels_v16_arXivSub1.tex
\documentclass[11pt, a4paper,onecolumn]{belarticle4}
\pdfoutput=1
\input{preamble.tex}

\usepackage{paralist}
\usepackage{subfigure}

\usepackage{tikz}

\definecolor{orangy}{RGB}{213,94,0}

\newcommand{\intr}{\mathrm{int}} 
\newcommand{\lspan}{\mathrm{span}}

\newcommand{\Aff}{\mathrm{Aff}} 
\newcommand{\DP}{\mathrm{DP}}
\newcommand{\PM}{\mathrm{MP}} 
\newcommand{\PMcone}{\mathcal{K}} 
\newcommand{\core}{\mathrm{core}} 
\newcommand{\LST}{\mathcal{L}(V_S,V_T)} 
\newcommand{\TST}{\mathcal{T}_S^T}

\newcommand{\discard}{\resizebox{4mm}{!}{%
\InputIfFileExists{Diagrams/discard.tikz}{}{\input{./figures/Diagrams/discard.tikz}}}}

\usepackage[numbers,sort&compress]{natbib}
\usepackage{mathdots}
\makeatletter
\DeclareRobustCommand{\rvdots}{
  \vbox{
    \baselineskip4\p@\lineskiplimit\z@
    \kern-\p@
    \hbox{.}\hbox{.}\hbox{.}
  }}
\makeatother
\allowdisplaybreaks
\begin{document}

\title{Simulating all multipartite non-signalling channels via quasiprobabilistic mixtures of local channels in  generalised probabilistic theories}
\author{Paulo J.~Cavalcanti}
\affiliation{International Centre for Theory of Quantum Technologies, University of Gda\'nsk, 80-309 Gda\'nsk, Poland}
\author{John H.~Selby}
\affiliation{International Centre for Theory of Quantum Technologies, University of Gda\'nsk, 80-309 Gda\'nsk, Poland}
\author{Jamie Sikora}
\affiliation{Virginia Polytechnic Institute and State University, Blacksburg, VA 24061, USA} 
\author{Ana Bel\'en Sainz}
\affiliation{International Centre for Theory of Quantum Technologies, University of Gda\'nsk, 80-309 Gda\'nsk, Poland}

\begin{abstract}
Non-signalling quantum channels -- relevant in, e.g., the study of Bell and  Einstein-Podolsky-Rosen  scenarios -- may be simulated via affine combinations of local operations in bipartite scenarios. Moreover, when these channels correspond to stochastic maps between classical variables, such simulation is possible even in multipartite scenarios. These two results have proven useful when studying the properties of these channels, such as their communication and information processing power, and even when defining measures of the non-classicality of physical phenomena (such as Bell non-classicality and steering). 
In this paper we show that such useful quasi-stochastic characterizations of channels may be unified and applied to the broader class of multipartite non-signalling channels. Moreover, we show that this holds for non-signalling channels in quantum theory, as well as in a larger family of generalised probabilistic theories. More precisely, we prove that non-signalling channels can always be simulated by affine combinations of corresponding local operations, provided that the underlying physical theory is locally tomographic -- a property that quantum theory satisfies. Our results then can be viewed as a generalisation of Refs.~[Phys. Rev. Lett. 111, 170403] and [Phys. Rev. A 88, 022318 (2013)] to the multipartite scenario for arbitrary tomographically local generalised probabilistic theories (including quantum theory). Our proof technique leverages Hardy's duotensor formalism, highlighting its utility in this line of research.
\end{abstract}

\maketitle

\section{Introduction}

Quantum operations are at the core of communication and information processing tasks, and how well we can perform at the latter may depend on the properties of the quantum operations that we have at hand. One particular set of operations of interest is that of \textit{non-signalling quantum channels} \cite{beckman2001causal}, i.e., those that cannot be used by two distant parties to exchange information in a way that is against the laws of relativity theory.  Bipartite non-signalling quantum operations have been extensively studied, specially since they play a central role in Bell \cite{Bell64} and Einstein-Podolsky-Rosen `steering' \cite{einstein1935can,schrodinger1935discussion} scenarios, which in turn underpin cryptographic protocols \cite{BellRev,SteeRev}. In addition, the simulation of bipartite non-signalling  quantum channels via affine combinations of local operations has provided valuable insight on the exploration of the advantage they provide for communication and information processing tasks \cite{al2013simulating,chiribella2013quantum}.  

In recent years it has become fruitful to study quantum theory from the `outside', that is, by placing it as one theory within a broad landscape of logically consistent theories. This allows one to understand \emph{why} quantum theory has particular features,  and also its  possibilities and limitations for various applications.
The framework of generalised probabilistic theories \cite{hardy2001quantum,barrett2007gpt} (GPTs) has become the preferred tool for such studies, for example, shedding light on matters pertaining to: cryptography 
	\cite{
	sikora2018simple,
	selby2018make,
	sikora2019impossibility,
	barnum2008nonclassicality,
	lami2018ultimate,
	barrett2005no}; 
computation \cite{
	barrett2017computational,
	lee2015computation,
	krumm2018quantum,
	lee2016generalised,
	lee2016deriving,
	barnum2018oracles,
	muller2012structure}; 
interference \cite{
	ududec2011three,
	lee2017higher,
	garner2018interferometric,
	barnum2014higher,
	dakic2014density,
	barnum2017ruling,
	horvat2020interference}; 
thermodynamics \cite{
	short2010entropy,
	barnum2010entropy,
	chiribella2015entanglement,
	kimura2016entropies,
	chiribella2017microcanonical,
	barnum2015entropy};
contextuality \cite{
	schmid2020structure,
	schmid2019characterization,
	shahandeh2019contextuality,
	chiribella2014measurement};
nonlocality \cite{
	gross2010all,
	czekaj2018bell,
	barnum2010local,
	czekaj2020correlations,
	henson2014theory,
	weilenmann2020analysing,
	lami2018non};
steering \cite{
	barnum2013ensemble,
	plavala2017conditions,
	banik2015measurement,
	cavalcanti2021witworld}; 
decoherence \cite{
	richens2017entanglement,
	lee2018no,
	scandolo2018possible,
	selby2017leaks};
information processing \cite{
	bae2016structure,
	barrett2007information,
	JP17,
	barnum2007generalized,
	barnum2011information,
	barnum2012teleportation,
	heinosaari2019no}; 
incompatibility \cite{
	plavala2016all,
	jenvcova2018incompatible,
	filippov2017necessary,
	kuramochi2020compatibility,
	bluhm2020incompatibility};
uncertainty \cite{
	dahlsten2014uncertainty,
	takakura2020preparation,
	takakura2021entropic,
	sun2020no};
as well as providing a foundational view of the primitive structures of physical theories \cite{
	masanes2019measurement,
	galley2017classification,
	galley2018any,
	galley2020dynamics,
	chiribella2011informational,
	chiribella2014dilation,
	chiribella2014distinguishability,
	barnum2014local,
	branford2018defining,
	wilce2009four,
	wilce2018royal,
	barnum2016composites,
	barnum2013symmetry,
	wilce2011symmetry,
	masanes2011derivation,
	masanes2013existence,
	mueller2013three}. For a  comprehensive introduction to the field see Refs.~\cite{lami2018non,plavala2021general,muller2021probabilistic}.

In this work we investigate no-signalling channels in GPTs. In particular, we prove a useful technical result, namely
that multipartite non-signalling channels in locally-tomographic GPTs \cite{hardy2001quantum} can be simulated by affine combinations of product (local) channels (Theorem \ref{thm:main}).  Our results can be viewed as a generalisation of those of Ref.~\cite{al2013simulating} and of Ref.~\cite[Lem.~1]{chiribella2013quantum} to arbitrary tomographically local GPTs: the former applies only to multipartite non-signalling stochastic maps on classical variables, while the latter applies to bipartite non-signalling quantum channels. 

Our proofs leverage the convenient duotensor formalism of Ref.~\cite{hardy2011reformulating} with a slight twist based on Ref.~\cite{van2017quantum} which allows us to directly lift the result of Ref.~\cite{al2013simulating} (using a generalisation of Lem.~2 in Ref.~\cite{chiribella2013quantum}) to this more general setting. We believe that this way of lifting structural properties of stochastic maps to properties of channels in arbitrary tomographically local GPTs via the duotensor formalism \cite{hardy2011reformulating} may be a useful tool in future research.

\section{Generalised Probabilistic Theories: the basics}

The framework of Generalised Probabilistic Theories (GPTs) can be used to define arbitrary physical theories. The simplicity of the framework enables various alternative theories to be formulated and explored while allowing at the same time a deep study of the probabilistic and compositional aspects of such theories. It is based on the tenet that a minimal requirement of any physical theory is that it must make probabilistic predictions about the outcomes of experiments. Whilst this is conceptually extremely minimal, the mathematical consequences of this lead to a rich formal structure known as a GPT.

Because physical theories describe predictions about measurement outcomes in experiments, a few elements are necessarily present in all of them. Namely, these theories need to talk about 
types of systems, possible states for each of them, possible measurement outcomes, transformations, and the operation of discarding a system (see Table \ref{table:gpts}). In quantum theory, these elements are, respectively, the Hilbert spaces, the density operators on them, positive operators upper-bounded by the identity, completely positive trace-non-increasing (CPTNI) linear maps, and the (partial) trace operation.

\begin{table}
\begin{center}
\begin{tabular}{| l || c | c | c | c}
\hline
Elements of a GPT &  Quantum theory & Classical theory \\ \hline
Systems &  Hilbert spaces & finite sets \\
States &  density operators  & probability distributions\\
Effects &  POVM element & $[0,1]-$valued functions \\
Discarding effect &  (partial) trace & marginalisation\\
Transformations &  CPTNI linear maps & substochastic maps \\
Composition rule &  tensor product &  Cartesian product \\ \hline
\end{tabular}
\end{center}
\caption{Elements that define a generalised probabilistic theory, and how they are defined for the particular case of quantum and classical theories viewed as GPTs. }
\label{table:gpts}
\end{table}

Having those elements present, although necessary, is not sufficient to express the full form of a physical theory. Some structure relating them are implied by the way that experiments are performed. Abstractly speaking, a notion of connectivity between those elements must also be present because, in experiments, we perform actions on systems, that is, we subject them to processes, and these processes can happen in parallel (independently) or in sequence. This motivates a notion of compositionality of processes.

From this notion of how the experimental processes connect, or compose, a convenient diagrammatic notation can be defined so as to capture the entire structure of the GPTs. We can represent any process by a box, and encode the type of system on which it happens as a labelled input wire at the bottom of it. (Hence, we have also implied that systems are represented by wires.) Additionally, since the type of a system can change after a process, we denote the output type of a process by a labelled wire on top of its box. In this notation, then, a system type $S$, and a transformation $T$ from a system type $A$ to a system type $B$, respectively, appear as

\begin{equation}\label{}
	\begin{tikzpicture}
		\node [style=none] (0) at (0, 0.75) {};
		\node [style=none] (1) at (0, -0.75) {};
		\node [style={right label}] (2) at (0, -0.5) {$S$};
		\draw (0.center) to (1.center);
	\end{tikzpicture}
	\quad\text{and }\quad 
	\begin{tikzpicture}
	\begin{pgfonlayer}{nodelayer}
		\node [style=none] (0) at (0, -1.25) {};
		\node [style={small box}] (1) at (0, -0) {$T$};
		\node [style={right label}] (2) at (0, -1.1) {$A$};
		\node [style={right label}] (3) at (0, 1) {$B$};
		\node [style=none] (4) at (0, 1.25) {};
	\end{pgfonlayer}
	\begin{pgfonlayer}{edgelayer}
		\draw (0.center) to (1);
		\draw (4.center) to (1);
	\end{pgfonlayer}
\end{tikzpicture}.
\end{equation}
Note that a state of a system can be conceptualised as some  \textit{preparation procedure}, which, abstractly speaking, is also a process. Hence we can represent it as a box that has no input wire but has as output the wire corresponding to the type of that system. Similarly, an effect, or measurement outcome, is a box with input wire corresponding to the system where it can be observed, and no output wire. We follow the convention that states and effects are represented by triangular boxes, so a state $\sigma$ and an effect $e$ of a system $S$ appear as
\begin{equation}\label{}	
\begin{tikzpicture}
	\begin{pgfonlayer}{nodelayer}
		\node [style=none] (0) at (-1, 0.75) {};
		\node [style=point] (1) at (-1, -0.5) {$\sigma$};
		\node [style=right label] (6) at (-1, 0.25) {$S$};
	\end{pgfonlayer}
	\begin{pgfonlayer}{edgelayer}
		\draw  (0.center) to (1);
	\end{pgfonlayer}
\end{tikzpicture}
	\quad \text{and} \quad
\begin{tikzpicture}
	\begin{pgfonlayer}{nodelayer}
		\node [style=none] (0) at (-1, -0.75) {};
		\node [style=copoint] (1) at (-1, 0.5) {$e$};
		\node [style=right label] (6) at (-1, -0.4) {$S$};
	\end{pgfonlayer}
	\begin{pgfonlayer}{edgelayer}
		\draw (0.center) to (1);
	\end{pgfonlayer}
\end{tikzpicture},
\end{equation}
respectively. Because of this, the discarding operation, since it has an input but no output, appears as a special effect in the theory. This effect is sometimes called the deterministic effect and is unique for each system type\footnote{The uniqueness of this discarding effect means that we are dealing with so-called \emph{causal} GPTs \cite{chiribella2010probabilistic}.}. In this notation, it is represented by
\begin{equation}\label{}	
\begin{tikzpicture}
	\begin{pgfonlayer}{nodelayer}
		\node [style=none] (0) at (-1, -0.75) {};
		\node [style=none] (1) at (-1, 0.25) {};
		\node [style=right label] (6) at (-1, -0.4) {$S$};
		\node [style=upground] (7) at (-1, 0.5) {};
	\end{pgfonlayer}
	\begin{pgfonlayer}{edgelayer}
		\draw (0.center) to (1.center);
	\end{pgfonlayer}
\end{tikzpicture}.
\end{equation}
These diagrammatic pieces can be connected when the input/output wire types match. This represents the sequential composition of processes. When processes are instead drawn side by side, we are representing their parallel composition. By connecting boxes, therefore, we can then construct more complex diagrams, i.e. complex processes, such as
\begin{equation}\label{}
\InputIfFileExists{Diagrams/connections.tikz}{}{\input{./figures/Diagrams/connections.tikz}}\ ,
\end{equation}
where we omit the wire labels for simplicity, but it should be clear that only matching types can be connected.

When a diagram has no loose wires, they are interpreted as numbers, which in the case of GPTs are the probabilities generated by the theory. For instance,
\begin{equation}\label{}
\InputIfFileExists{Diagrams/probability.tikz}{}{\input{./figures/Diagrams/probability.tikz}}\quad =\ \  \mathsf{Prob}(e|g,\sigma)
\end{equation}
 denotes the probability that the outcome associated to effect $e$ is observed when the system is prepared in state $\sigma$ and a transformation $g$ is applied to it. 

Of course, we might need to describe systems that are composed by simpler parts -- multipartite systems -- so we can emphasize that some system is composite by drawing the wires of its parts side by side
\begin{equation}\label{}
\InputIfFileExists{Diagrams/composite1.tikz}{}{\input{./figures/Diagrams/composite1.tikz}}\quad =\quad %
\InputIfFileExists{Diagrams/composite2.tikz}{}{\input{./figures/Diagrams/composite2.tikz}}.
\end{equation}
When we represent bipartite composite systems by the two wires together, its deterministic effect is represented by the composition of the deterministic effects of its parts:
\begin{equation}\label{}
\InputIfFileExists{Diagrams/deteffect1.tikz}{}{\input{./figures/Diagrams/deteffect1.tikz}}\quad =\quad %
\InputIfFileExists{Diagrams/deteffect2.tikz}{}{\input{./figures/Diagrams/deteffect2.tikz}}.
\end{equation}
With what we have, we can represent simple experimental processes, composite processes, and probabilities of outcomes in those experiments. To reason about them, we need now a notion of equality of processes, or, in other words, a notion of tomography. We say that two processes are equal if they give the same probabilities in all situations, so
\begin{equation}\label{}
\InputIfFileExists{Diagrams/tomography1.tikz}{}{\input{./figures/Diagrams/tomography1.tikz}}\quad =\quad %
\InputIfFileExists{Diagrams/tomography2.tikz}{}{\input{./figures/Diagrams/tomography2.tikz}}\qquad \iff \qquad  %
\InputIfFileExists{Diagrams/tomography3.tikz}{}{\input{./figures/Diagrams/tomography3.tikz}}\quad =\quad %
\InputIfFileExists{Diagrams/tomography4.tikz}{}{\input{./figures/Diagrams/tomography4.tikz}} \quad \forall \sigma,e.
\end{equation}
In fact, in this paper we will work with a special class of GPTs which satisfy the principle of \emph{tomographic locality} \cite{hardy2001quantum,hardy2011reformulating}. This means that:
\begin{equation}\label{}
\InputIfFileExists{Diagrams/tomography1.tikz}{}{\input{./figures/Diagrams/tomography1.tikz}}\quad =\quad %
\InputIfFileExists{Diagrams/tomography2.tikz}{}{\input{./figures/Diagrams/tomography2.tikz}}\qquad \iff \qquad  %
\InputIfFileExists{Diagrams/tomography5.tikz}{}{\input{./figures/Diagrams/tomography5.tikz}}\quad = \quad%
\InputIfFileExists{Diagrams/tomography6.tikz}{}{\input{./figures/Diagrams/tomography6.tikz}} \quad \forall \sigma,e,
\end{equation}
that is, in a tomographically local theory we can do process tomography without a side channel.

An important type of processes is that of those that are \textit{discard-preserving}, which means they satisfy the following: 
\begin{equation}\label{}
\InputIfFileExists{Diagrams/causal1.tikz}{}{\input{./figures/Diagrams/causal1.tikz}}\quad =\quad %
\InputIfFileExists{Diagrams/causal2.tikz}{}{\input{./figures/Diagrams/causal2.tikz}}.
\end{equation}
In the case of quantum theory, these correspond to the trace-preserving maps. Physically-realisable discard-preserving processes in a GPT are known as \emph{channels}.
Discard-preservation also defines  a notion of causality for processes \cite{chiribella2010probabilistic,coecke2016terminality}: a process is said to be causal if it is discarding preserving.  This is so because this condition ensures compatibility with relativistic causal structure \cite{kissinger2017equivalence}.

A final ingredient in the GPT formalism is the possibility to represent convex mixtures of processes. This stems from the requirement that in an experiment we can always decide to perform $f$ with probability $p$ or $g$ with probability $(1-p)$, at least, provided that $f$ and $g$ have the same input and output systems. This is introduced through the definition of a sum of processes that distributes over diagrams:
\begin{equation}\label{}
\InputIfFileExists{Diagrams/mixture1.tikz}{}{\input{./figures/Diagrams/mixture1.tikz}}\quad =\quad %
\InputIfFileExists{Diagrams/mixture2.tikz}{}{\input{./figures/Diagrams/mixture2.tikz}}.
\end{equation}
Note that this definition implies that we can only sum processes with the same input/output types. From this, since a probability $p$ can be a number (diagram without loose wires) of the theory, we can write
\begin{equation}\label{}
\InputIfFileExists{Diagrams/mixture3.tikz}{}{\input{./figures/Diagrams/mixture3.tikz}}\quad =\quad p\ \ %
\InputIfFileExists{Diagrams/mixture4.tikz}{}{\input{./figures/Diagrams/mixture4.tikz}}\   + \ (1-p)\ \ %
\InputIfFileExists{Diagrams/mixture5.tikz}{}{\input{./figures/Diagrams/mixture5.tikz}}\quad =\quad \sum_{i} p_{i}\ %
\InputIfFileExists{Diagrams/mixture6.tikz}{}{\input{./figures/Diagrams/mixture6.tikz}}
\end{equation}
to describe convex mixtures of processes.

At this point a notion of order can be defined for processes:
\begin{equation}\label{}
\InputIfFileExists{Diagrams/order1.tikz}{}{\input{./figures/Diagrams/order1.tikz}}\quad \leq\quad %
\InputIfFileExists{Diagrams/order2.tikz}{}{\input{./figures/Diagrams/order2.tikz}}\qquad \iff\qquad \exists\  %
\InputIfFileExists{Diagrams/order3.tikz}{}{\input{./figures/Diagrams/order3.tikz}} \quad\text{s.t. }\ %
\InputIfFileExists{Diagrams/order1.tikz}{}{\input{./figures/Diagrams/order1.tikz}} +\ \ %
\InputIfFileExists{Diagrams/order3.tikz}{}{\input{./figures/Diagrams/order3.tikz}}\ =\ \ %
\InputIfFileExists{Diagrams/order2.tikz}{}{\input{./figures/Diagrams/order2.tikz}}.
\end{equation}
This order allows us to define discard-non-increasing processes. A process $f$ is said to be discard-non-increasing if and only if
\begin{equation}\label{}
\InputIfFileExists{Diagrams/causal1.tikz}{}{\input{./figures/Diagrams/causal1.tikz}}\quad \leq\quad %
\InputIfFileExists{Diagrams/causal2.tikz}{}{\input{./figures/Diagrams/causal2.tikz}}.
\end{equation}
Note that in any GPT all (physically-realisable) processes must be discard-nonincreasing; this corresponds to the constraint in quantum theory that processes are trace-nonincreasing.
In particular, this means that for any effect in the theory, there must be another effect such that they sum to the deterministic effect. Note that this is important for the definition of measurements. In quantum theory, for example, the deterministic effect is the trace operation, or multiplication by identity followed by the trace, and it is required that the POVM elements forming a measurement sum to identity, so each of them is less than or equal to the determistic effect.

Since in this work we focus on the class of GPTs that are tomographically local, we can moreover use the particular duotensor notation of Ref.~\cite{hardy2011reformulating}. Next we will present the basics of this notation.

\section{Duotensor basics}

Here we present an adaptation to the duotensor formalism where,
in addition to the GPT systems of the previous section, we also have classical systems representing measurement outcomes and control systems. In order to distinguish these two kinds of systems,  the classical ones will be drawn horizontally. We will also label them by finite sets, $\Lambda$:
\beq
\InputIfFileExists{Diagrams/fidInvId1New.tikz}{}{\input{./figures/Diagrams/fidInvId1New.tikz}}.
\eeq
The physical processes transforming between these classical systems are (sub)stochastic maps between these finite sets. We draw these as white boxes, such as:
\beq
\InputIfFileExists{Diagrams/stochMap.tikz}{}{\input{./figures/Diagrams/stochMap.tikz}}.
\eeq
A particularly useful example which we will make use of in this work is the copy map, which we draw as a white dot and is defined by:
\beq
\InputIfFileExists{Diagrams/copyDef1.tikz}{}{\input{./figures/Diagrams/copyDef1.tikz}} \quad = \quad %
\InputIfFileExists{Diagrams/copyDef2.tikz}{}{\input{./figures/Diagrams/copyDef2.tikz}} \qquad \forall \lambda \in\Lambda.
\eeq
Note that the copy map satisfies:
\beq\label{eq:compCopy}
\InputIfFileExists{Diagrams/copyComp1.tikz}{}{\input{./figures/Diagrams/copyComp1.tikz}} \quad = \quad %
\InputIfFileExists{Diagrams/copyComp2.tikz}{}{\input{./figures/Diagrams/copyComp2.tikz}},
\eeq
that is, copying the components of a system is the same as copying the composite system.

In contrast to the physical (sub)stochastic maps, we will draw mathematically well defined but unphysical processes as black boxes such as:
\beq
\InputIfFileExists{Diagrams/unphysMap.tikz}{}{\input{./figures/Diagrams/unphysMap.tikz}},
\eeq
which, in this case, would be a linear map from $\Lambda$ to $\Lambda'$ which is not (sub)stochastic, e.g., it may have negative coefficients.

 Note that in contrast to the approach of Ref.~\cite{hardy2011reformulating}, rather than labelling horizontal systems by black and white dots, we instead label the processes as being either black or white. This is equivalent but more convenient for us as, on the one hand, we can interpret the color as representing whether or not a process is physical, and, on the other hand, it takes us to a more standard category-theoretic notation. Indeed, categorically there is no distinction between the horizontal and vertical wires, it is simply a convenient way to label the different objects, at which point it is clear that all of the processes that we draw below live inside the category of real linear maps.

For each system $S$ in the GPT we define a particular minimal informationally-complete state preparation and measurement. We call these the fiducial preparation and fiducial measurement. A state preparation is a box which has a classical input and a GPT output where the classical input controls which state is prepared, whilst a measurement is a box which has a GPT input and a classical output where the classical output encodes the result of the measurement. We can therefore denote the fiducial preparation and fiducial measurement for a system $S$ as:
\beq
\InputIfFileExists{Diagrams/fidPrep.tikz}{}{\input{./figures/Diagrams/fidPrep.tikz}}\quad\text{and}\quad %
\InputIfFileExists{Diagrams/fidMeas.tikz}{}{\input{./figures/Diagrams/fidMeas.tikz}}
\eeq
where without loss of generality we take $\Lambda_S$ to index both the fiducial set of states and the fiducial set of effects. Moreover, all of the fiducial states are normalised and the fiducial effects sum to the unit effect, such that:
\beq \label{eq:Norm}
\InputIfFileExists{Diagrams/fidPrepNorm.tikz}{}{\input{./figures/Diagrams/fidPrepNorm.tikz}}\quad =\quad %
\InputIfFileExists{Diagrams/fidPrepNorm1.tikz}{}{\input{./figures/Diagrams/fidPrepNorm1.tikz}} \qquad \text{and} \qquad %
\InputIfFileExists{Diagrams/fidMeasNorm.tikz}{}{\input{./figures/Diagrams/fidMeasNorm.tikz}}\quad = \quad%
\InputIfFileExists{Diagrams/fidMeasNorm1.tikz}{}{\input{./figures/Diagrams/fidMeasNorm1.tikz}}.
\eeq
Note that here we follow the convention of Ref.~\cite{van2017quantum} rather than Ref.~\cite{hardy2011reformulating}, as the former demands that the fiducial effects form a measurement whilst the latter does not. Note that this does not constitute a loss of generality as a minimal informationally-complete measurement can be shown to exist for any GPT.

Now, for each system $S$, define the \emph{fiducial transition matrix} by 
\beq
\InputIfFileExists{Diagrams/fidTransition1.tikz}{}{\input{./figures/Diagrams/fidTransition1.tikz}}\quad :=\quad %
\InputIfFileExists{Diagrams/fidTransition.tikz}{}{\input{./figures/Diagrams/fidTransition.tikz}}
\eeq
and note that eq.~\eqref{eq:Norm} implies that these fiducial transition matrices are  stochastic maps, such that:
\beq
\InputIfFileExists{Diagrams/fidTransNorm.tikz}{}{\input{./figures/Diagrams/fidTransNorm.tikz}}\quad =\quad %
\InputIfFileExists{Diagrams/fidTransNorm1.tikz}{}{\input{./figures/Diagrams/fidTransNorm1.tikz}}.
\eeq

Now, the fact that the fiducial preparation and measurement are informationally-complete means that they are invertible linear maps. Importantly, however, these inverses are \emph{not} typically physical transformations. We therefore denote them as:
\beq
\InputIfFileExists{Diagrams/fidPrepInv.tikz}{}{\input{./figures/Diagrams/fidPrepInv.tikz}} \qquad \text{and} \qquad %
\InputIfFileExists{Diagrams/fidMeasInv.tikz}{}{\input{./figures/Diagrams/fidMeasInv.tikz}}
\eeq
such that:
\beq\label{eq:fidInv}
\InputIfFileExists{Diagrams/fidInv1.tikz}{}{\input{./figures/Diagrams/fidInv1.tikz}}\ \ =\ \ %
\InputIfFileExists{Diagrams/fidInvId1.tikz}{}{\input{./figures/Diagrams/fidInvId1.tikz}} \ \ =\ \ %
\InputIfFileExists{Diagrams/fidInv2.tikz}{}{\input{./figures/Diagrams/fidInv2.tikz}} \qquad \text{and} \qquad %
\InputIfFileExists{Diagrams/fidInv3.tikz}{}{\input{./figures/Diagrams/fidInv3.tikz}} \ \ =\ \ %
\InputIfFileExists{Diagrams/fidInvId3.tikz}{}{\input{./figures/Diagrams/fidInvId3.tikz}} \ \ =\ \ %
\InputIfFileExists{Diagrams/fidInv4.tikz}{}{\input{./figures/Diagrams/fidInv4.tikz}}.
\eeq
We can then moreover define
\beq
\InputIfFileExists{Diagrams/fidTransInv1.tikz}{}{\input{./figures/Diagrams/fidTransInv1.tikz}}\quad :=\quad %
\InputIfFileExists{Diagrams/fidTransInv.tikz}{}{\input{./figures/Diagrams/fidTransInv.tikz}},
\eeq
which can easily be seen using Eq.~\eqref{eq:fidInv} to be the inverse of the fiducial transition matrix. Hence:
\beq
\InputIfFileExists{Diagrams/fidTransInv2.tikz}{}{\input{./figures/Diagrams/fidTransInv2.tikz}} \quad =\quad  %
\InputIfFileExists{Diagrams/fidTransInv3.tikz}{}{\input{./figures/Diagrams/fidTransInv3.tikz}}\quad  =\quad  %
\InputIfFileExists{Diagrams/fidTransInv4.tikz}{}{\input{./figures/Diagrams/fidTransInv4.tikz}}.   
\eeq
The fiducial transition matrix and its inverse (the white and black squares respectively) are known as \emph{hopping metrics} in the terminology of Hardy.

It is also easy to see from these conditions that:
\beq
\InputIfFileExists{Diagrams/fidPrepInv.tikz}{}{\input{./figures/Diagrams/fidPrepInv.tikz}}\quad  =\quad %
\InputIfFileExists{Diagrams/fidTransInv6.tikz}{}{\input{./figures/Diagrams/fidTransInv6.tikz}}  \qquad \text{and} \qquad %
\InputIfFileExists{Diagrams/fidMeasInv.tikz}{}{\input{./figures/Diagrams/fidMeasInv.tikz}} \quad =\quad  %
\InputIfFileExists{Diagrams/fidTransInv5.tikz}{}{\input{./figures/Diagrams/fidTransInv5.tikz}},
\eeq
which, in particular, means that:
\beq\label{eq:IdDecomp}
\InputIfFileExists{Diagrams/IdDecomp1.tikz}{}{\input{./figures/Diagrams/IdDecomp1.tikz}}\quad  =\quad  %
\InputIfFileExists{Diagrams/IdDecomp.tikz}{}{\input{./figures/Diagrams/IdDecomp.tikz}}.
\eeq
Moreover, it is also easy to show that:
\beq \label{eq:BlackNorm}
\InputIfFileExists{Diagrams/fidPrepNormB.tikz}{}{\input{./figures/Diagrams/fidPrepNormB.tikz}}\quad  =\quad  %
\InputIfFileExists{Diagrams/fidPrepNorm1.tikz}{}{\input{./figures/Diagrams/fidPrepNorm1.tikz}} \ \text{,} \qquad %
\InputIfFileExists{Diagrams/fidMeasNormB.tikz}{}{\input{./figures/Diagrams/fidMeasNormB.tikz}}\quad  =\quad  %
\InputIfFileExists{Diagrams/fidMeasNorm1.tikz}{}{\input{./figures/Diagrams/fidMeasNorm1.tikz}} \  \text{, \quad and}\qquad %
\InputIfFileExists{Diagrams/fidTransNormB.tikz}{}{\input{./figures/Diagrams/fidTransNormB.tikz}}\quad  =\quad  %
\InputIfFileExists{Diagrams/fidTransNorm1.tikz}{}{\input{./figures/Diagrams/fidTransNorm1.tikz}}\ .
\eeq

The key use of all of this for us, is that it allows us to map  any GPT channel to a stochastic map and back again as follows: a GPT channel is mapped to a stochastic map via
\beq \label{eq:themapGPTstoc}
\InputIfFileExists{Diagrams/GPTChannel.tikz}{}{\input{./figures/Diagrams/GPTChannel.tikz}}\quad  \mapsto\quad  %
\InputIfFileExists{Diagrams/GPTChannel2.tikz}{}{\input{./figures/Diagrams/GPTChannel2.tikz}}\,,
\eeq
and the stochastic map associated to the GPT channel can be mapped back to the GPT channel via
\beq
\InputIfFileExists{Diagrams/GPTChannel2.tikz}{}{\input{./figures/Diagrams/GPTChannel2.tikz}}\quad  \mapsto \quad  %
\InputIfFileExists{Diagrams/GPTChannel3.tikz}{}{\input{./figures/Diagrams/GPTChannel3.tikz}}\quad  =\quad   %
\InputIfFileExists{Diagrams/GPTChannel.tikz}{}{\input{./figures/Diagrams/GPTChannel.tikz}}\,.
\eeq
It is clear that the RHS of Eq.~\eqref{eq:themapGPTstoc} is indeed stochastic as it is positive (since it is composed out of physically realisable GPT transformations) and satisfies:
\beq
\InputIfFileExists{Diagrams/GPTStochNorm1.tikz}{}{\input{./figures/Diagrams/GPTStochNorm1.tikz}} \quad =\quad  %
\InputIfFileExists{Diagrams/GPTStochNorm2.tikz}{}{\input{./figures/Diagrams/GPTStochNorm2.tikz}}\quad  =\quad  %
\InputIfFileExists{Diagrams/GPTStochNorm3.tikz}{}{\input{./figures/Diagrams/GPTStochNorm3.tikz}}\quad  =\quad  %
\InputIfFileExists{Diagrams/GPTStochNorm4.tikz}{}{\input{./figures/Diagrams/GPTStochNorm4.tikz}}
\eeq
where the second equality holds because $C$ is a GPT channel rather than a generic GPT process.
Similar arguments imply that if $C$ satisfies certain no-signalling conditions then so to will the associated stochastic map.

For example, a bipartite channel $B$ is said to be non-signalling if:
\beq
\InputIfFileExists{Diagrams/NSChannelDef1.tikz}{}{\input{./figures/Diagrams/NSChannelDef1.tikz}}\quad  =\quad  %
\InputIfFileExists{Diagrams/NSChannelDef2.tikz}{}{\input{./figures/Diagrams/NSChannelDef2.tikz}} \qquad \text{and} \qquad %
\InputIfFileExists{Diagrams/NSChannelDef3.tikz}{}{\input{./figures/Diagrams/NSChannelDef3.tikz}}\quad  =\quad  %
\InputIfFileExists{Diagrams/NSChannelDef4.tikz}{}{\input{./figures/Diagrams/NSChannelDef4.tikz}}
\eeq
from which it is easy to show that the associated stochastic map will also be non-signalling, for example:
\beq
\InputIfFileExists{Diagrams/NSChannel1.tikz}{}{\input{./figures/Diagrams/NSChannel1.tikz}}\quad =\quad %
\InputIfFileExists{Diagrams/NSChannel2.tikz}{}{\input{./figures/Diagrams/NSChannel2.tikz}}\quad =\quad %
\InputIfFileExists{Diagrams/NSChannel3.tikz}{}{\input{./figures/Diagrams/NSChannel3.tikz}}\quad =\quad %
\InputIfFileExists{Diagrams/NSChannel4.tikz}{}{\input{./figures/Diagrams/NSChannel4.tikz}}.
\eeq
This straightforwardly generalises to  multipartite GPT channels, and also to the case where only some of the no-signalling conditions hold. That is, the non-signalling structure of the channel and of the associated stochastic map are the same.

\section{Geometry of transformations}\label{section:geometry}

 In this section we present a geometric perspective on some of the processes discussed above, as well as on particular types of channels.

\subsection{States}
First let us start by discussing the geometry of the state space for some system $S$. Schematically this looks like:
\beq
\includegraphics[clip, trim=1cm 13.5cm 1cm 5cm,scale=.5]{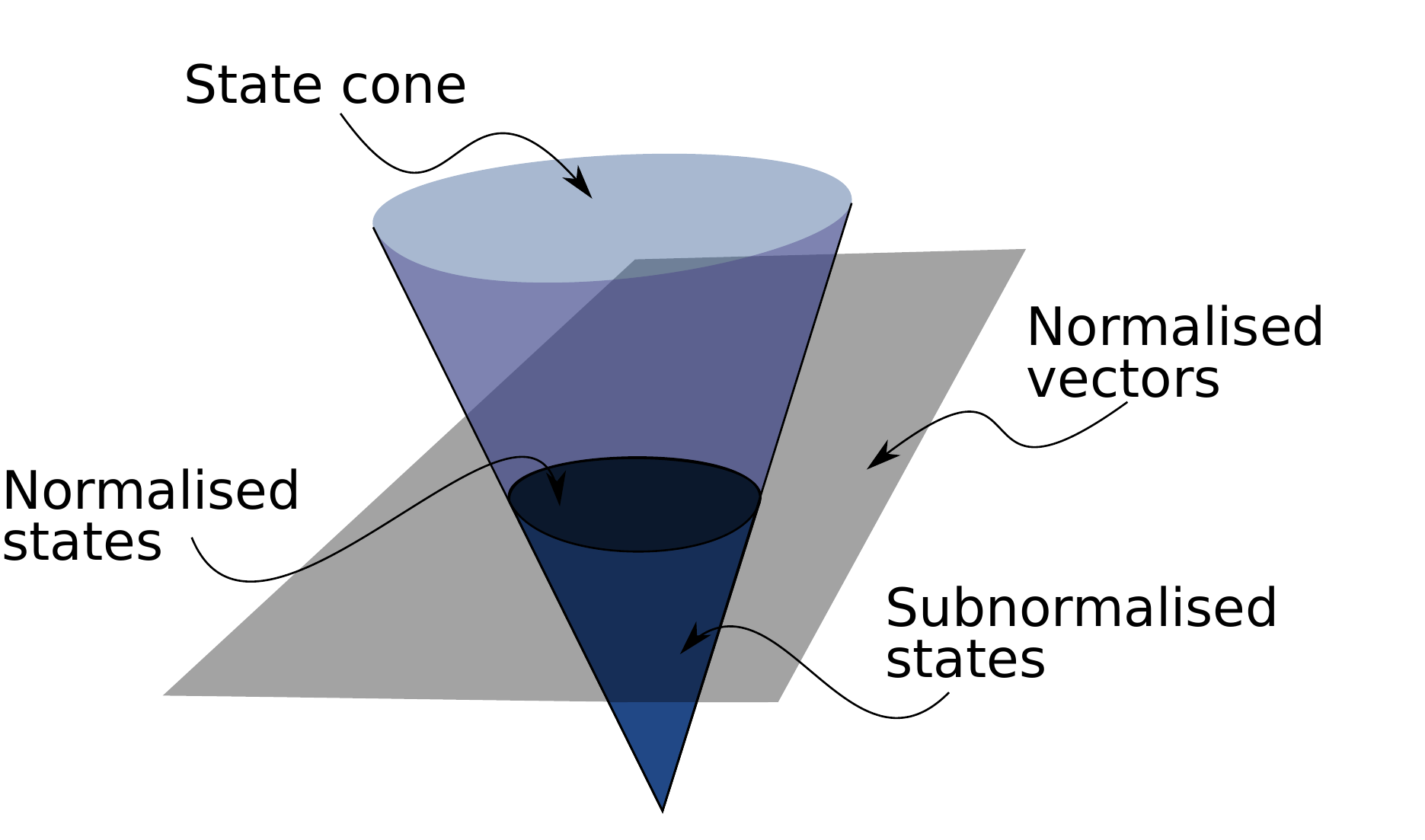}
\eeq

Formally, we have some real vector space $V_S$ which contains a convex cone of states, $\mathcal{T}^S$, which is closed, pointed, and full dimensional, with an intersecting hyperplane which defines the normalised vectors. The intersection of this hyperplane and the state cone defines the normalised state space, $\Omega_S$.  
A subnormalized state $s$ is a vector in the cone such that there exists $\alpha \geq 1$ such that $\alpha s$ is normalized. 
In particular, the convex set of subnormalised states spans the vector space and, moreover, there exists at least one normalized state which is interior to the cone.

\subsection{Effects}\label{section:effects}
Next let us consider the geometry of the effect space for some system $S$. Schematically this looks like:
\beq
\includegraphics[clip, trim=1cm 16.5cm 1cm 4cm,scale=.5]{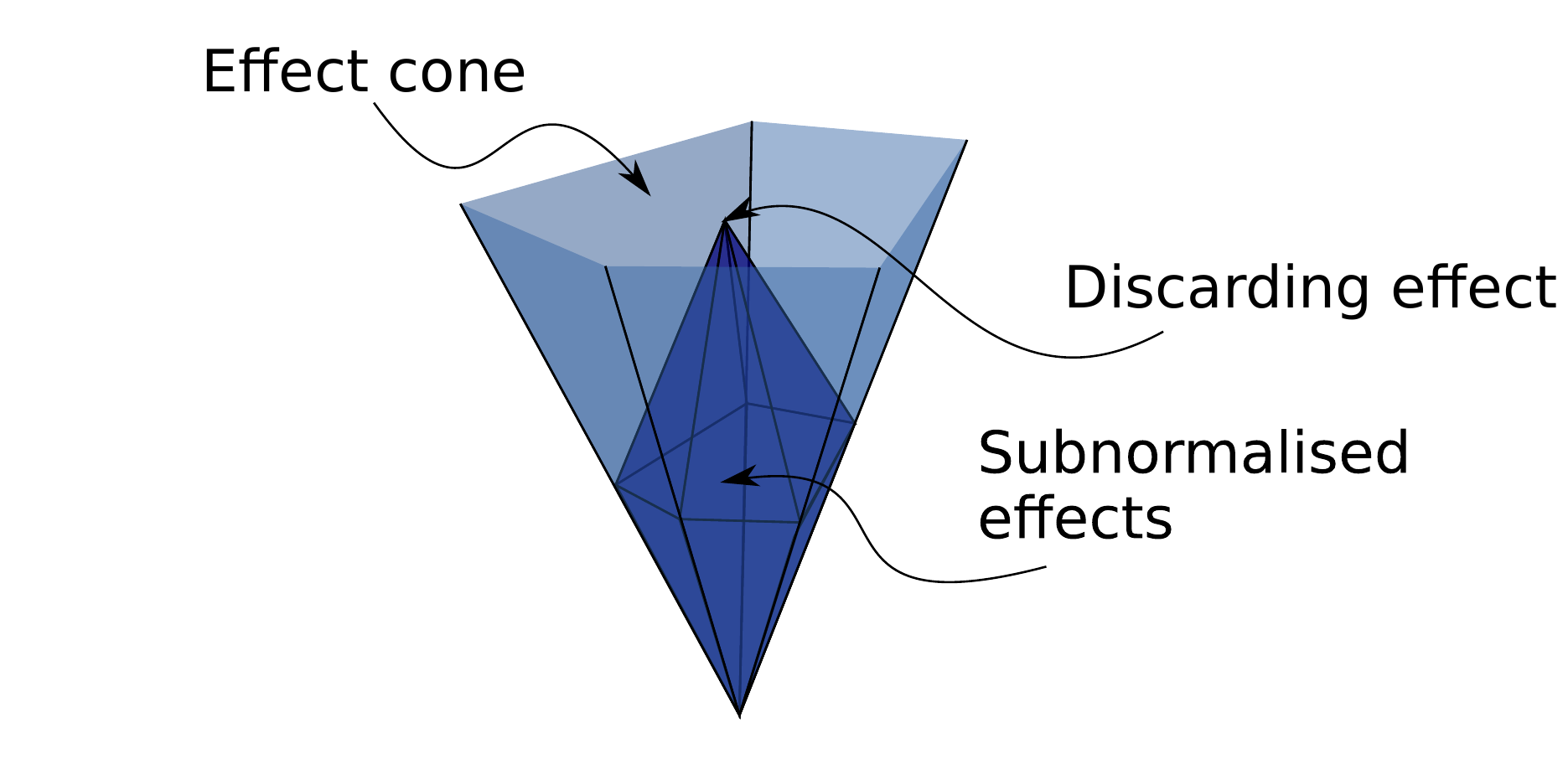}
\eeq

Formally, the effect space of $S$ lives inside the dual of the vector space of states, $V_S^*$, and consists of a convex cone of effects, $\mathcal{T}_S$, which is closed, pointed and full dimensional. 
The unique ``normalised'' effect (the \emph{discarding effect}) $\discard$,  is the unique linear functional that evaluates to $1$ on the intersecting hyperplane defining the normalised states. 
This must be in the interior of the effect cone such that it is an order unit for the cone. 
That is, we have that every effect $e$ in the cone can be rescaled to an effect $\alpha e$, for some $\alpha > 0$, such that there exists some $e'$ in the cone which satisfies $\alpha e + e'= \discard$. The set of subnormalised effects, $\mathcal{E}_S$, can be defined as those that satisfy this condition for some $\alpha \geq 1$. 
In particular, this ensures that the convex set of subnormalised effects spans the dual space and that $\discard$ is in the interior of the effect cone. 

\subsection{Physical transformations} 
Finally, we turn to our main focus which is the geometry of transformations within a tomographically local GPT. Schematically this looks like:
\beq
\includegraphics[clip, trim=0.25cm 16cm 1cm 4cm,scale=.5]{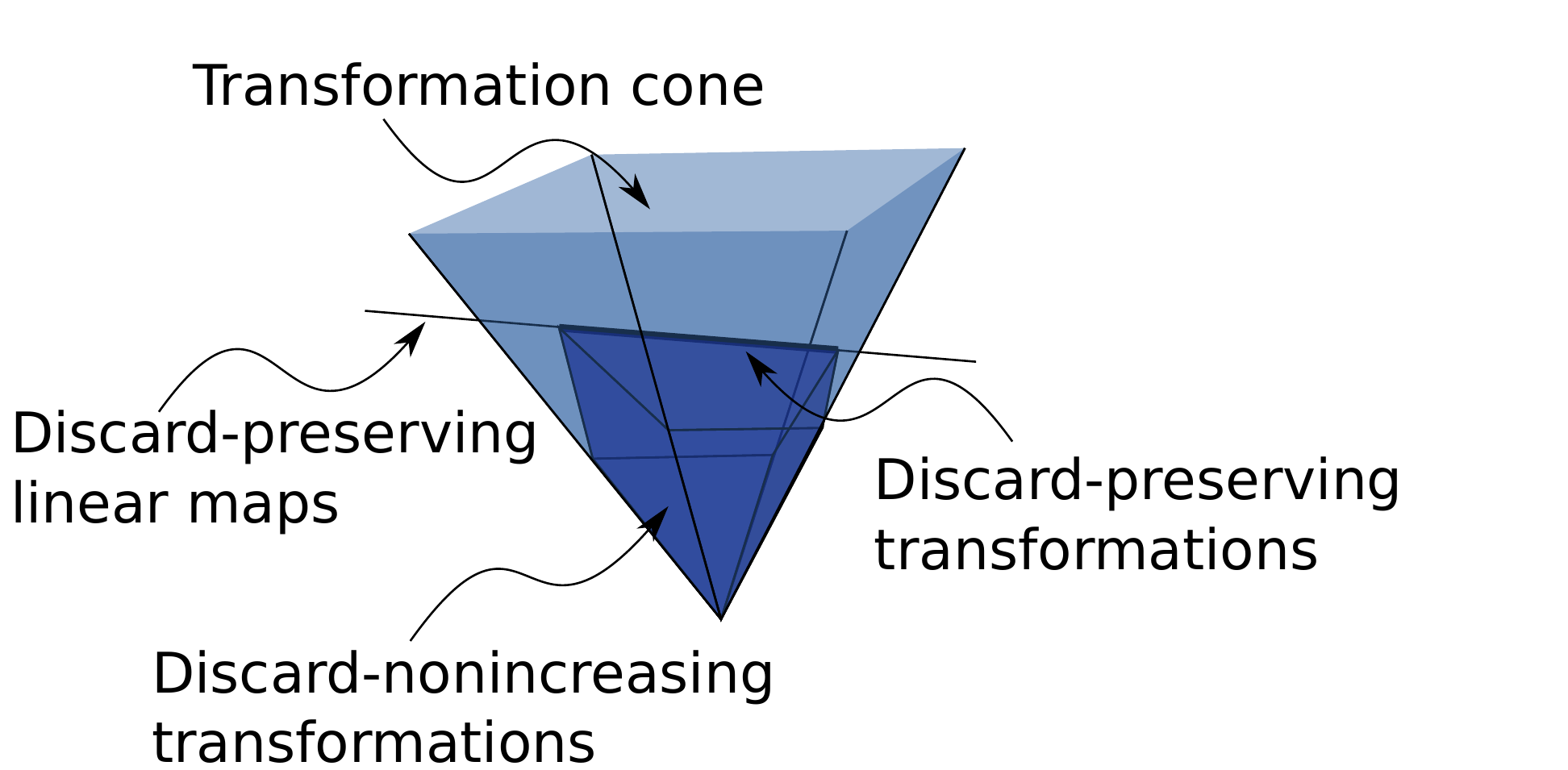}
\eeq

As we are assuming tomographic locality, the transformations from $S$ to $T$ live inside the vector space of linear maps from $V_S$ to $V_T$, which we denote as $\LST$. 
The geometric picture that we present here is not as standard in the literature as it is for the state and effect cases, and so we  now explain how this structure arises. 

In this picture we have a convex set of normalised transformations which are defined by the intersection of an affine  set (namely, the discard-preserving linear maps) and a convex cone (namely, the cone of transformations, $\TST$). 
We can then view this as a positive cone such that the discard-nonincreasing transformations are those that are ``underneath'' the discard-preserving transformations in the associated partial order. 

As there exists a set of states which span $V_T$ and a set of effects which span $V_S^*$, then, using the fact that $\mathcal{L}(V_S,V_T) \cong V_S^* \otimes V_T$, we have that 
\beq 
\LST \cong \lspan \left\{ 
\InputIfFileExists{Diagrams/effectstate.tikz}{}{\input{./figures/Diagrams/effectstate.tikz}}\ \middle| \ s \in \Omega_T, e \in \mathcal{T}_S
\right\}. 
\eeq 
This means that any linear map in $\LST$ can be written as 
\beq
\InputIfFileExists{Diagrams/effectstate3.tikz}{}{\input{./figures/Diagrams/effectstate3.tikz}} 
\eeq
for some finite values of $k$ and $k'$, where $s_1, \ldots, s_k, s'_1, \ldots, s'_{k'}$ are normalised states, and $e_1, \ldots, e_k, e'_1, \ldots, e'_{k'}$ are in the effect cone. 

With this in mind, we define the following convex subcone $\PMcone$ of the cone of transformations $\TST$ which is useful in our analysis:
\beq \label{Kcone}
\PMcone := \left\{ \sum_{i=1}^k
\InputIfFileExists{Diagrams/effectstate4.tikz}{}{\input{./figures/Diagrams/effectstate4.tikz}} \middle| s_i \in  \Omega_T, e_i \in \mathcal{T}_S, k \text{ finite} \right\} \subseteq \TST \subset \, \LST. 
\eeq 
Equipped with this definition, we can express $\LST$ neatly as 
\beq 
\LST = \PMcone - \PMcone  := \{ \phi_1 - \phi_2 | \phi_1, \phi_2 \in \PMcone \}, 
\eeq  
which means that $\PMcone$ spans $\LST$.

\subsection{Measure-and-prepare transformations and discard-preserving channels} 

Of particular interest is the set of physical transformations referred to as \emph{measure-and-prepare}, which we denote as $\PM$:
\beq
\PM := \left\{ 
\InputIfFileExists{Diagrams/effectstate2.tikz}{}{\input{./figures/Diagrams/effectstate2.tikz}} \middle| s_1, \ldots, s_k \in \Omega_T, e_1, \ldots, e_k \in \mathcal{E}_S, \sum_{i=1}^k e_i = \discard, k \text{ finite} \right\},  
\eeq 
recalling that $\discard$ is the discarding effect. 

Note that since these are physically possible in any GPT, they are a subset\footnote{In quantum theory, for example, these are a proper subset of all quantum channels and are known as \emph{entanglement-breaking channels.}} of the valid transformations, that is, they live inside the convex cone $\mathcal{T}_S^T$ and, in fact, $\PM \subseteq \PMcone$. 

A measure-and-prepare transformation $\phi\in \PM$ has the additional property of being discard-preserving:
\beq
\InputIfFileExists{Diagrams/causalphi.tikz}{}{\input{./figures/Diagrams/causalphi.tikz}}\quad =\quad \sum_{i=1}^k %
\InputIfFileExists{Diagrams/MPDiscard.tikz}{}{\input{./figures/Diagrams/MPDiscard.tikz}}\quad =\quad \sum_{i=1}^k %
\InputIfFileExists{Diagrams/MPDiscard2.tikz}{}{\input{./figures/Diagrams/MPDiscard2.tikz}}\quad =\quad %
\InputIfFileExists{Diagrams/MPDiscard3.tikz}{}{\input{./figures/Diagrams/MPDiscard3.tikz}}.
\eeq

We denote the set of discard-preserving linear maps as $\DP$ and define it formally as 
\beq 
\DP := \left\{ %
\InputIfFileExists{Diagrams/DP.tikz}{}{\input{./figures/Diagrams/DP.tikz}}\ \middle| 
\InputIfFileExists{Diagrams/causal2.tikz}{}{\input{./figures/Diagrams/causal2.tikz}} = %
\InputIfFileExists{Diagrams/causal1.tikz}{}{\input{./figures/Diagrams/causal1.tikz}} 
\right\}. 
\eeq
Note that the set of discard-preserving maps forms an affine space (see Appendix~\ref{AppBatman} for definitions relating to ``affine'' concepts), which is easily proven given the definition above. 

Note that $\DP$ may contain non-physical transformations. 
However, from the above discussion, it does contain the measure-and-prepare transformations. 
Neatly, we have that $\PM \subseteq \DP \cap \PMcone$. 
Perhaps surprisingly, this containment is not strict, as shown in Lemma \ref{lem:newlabel} (see the Appendix).

The sets $\DP$, $\PM$, and $\PMcone$ allow us to get a useful characterization of the discard-preserving linear maps, which we now discuss.

\subsection{A useful characterization of discard-preserving linear maps}

The following theorem characterizes the set of discard-preserving maps in terms of those that are also measure-and-prepare. 

\begin{theorem}\label{theorem:AffSpanTransf}
Any discard-preserving linear map can be written as an affine combination of measure-and-prepare transformations and any affine combination of measure-and-prepare transformations is a discard-preserving linear map. More formally,

\begin{equation}
    \DP = \Aff(\PM),
\end{equation}
where $\Aff$ denotes the affine hull operation.
\end{theorem}
We provide a proof of Theorem~\ref{theorem:AffSpanTransf} preceded by a background on convex geometry in Appendix~\ref{AppBatman}.

\section{A characterisation of no-signalling GPT channels}

Critical to our result is that of Ref.~\cite{al2013simulating}. In the duotensor formalism presented in the previous section, the result of Ref.~\cite{al2013simulating} is:
if we have some non-signalling stochastic map $S$ it can be written as an affine combination of product stochastic maps:
\beq
\InputIfFileExists{Diagrams/Proof1.tikz}{}{\input{./figures/Diagrams/Proof1.tikz}}\quad =\quad \sum_{\alpha \in A} q_\alpha\  %
\InputIfFileExists{Diagrams/Proof2.tikz}{}{\input{./figures/Diagrams/Proof2.tikz}}
 \eeq 
where $n$ is the number of input/output system, $q_\alpha \in \mathds{R}$, $\sum_{\alpha \in A} q_\alpha =1$, and the $s_i^\alpha$ are stochastic maps. That is the $q_\alpha$ define a quasiprobability distribution $q$ over the set $A$. We can therefore equivalently draw this as:
\beq
\InputIfFileExists{Diagrams/Proof1.tikz}{}{\input{./figures/Diagrams/Proof1.tikz}}\quad =\quad %
\InputIfFileExists{Diagrams/affineCombination.tikz}{}{\input{./figures/Diagrams/affineCombination.tikz}},
\eeq
where the white dot is the copy operation, the quasiprobability distribution $q$ is a black triangle because it is not physically realisable as it can have negative coefficients, and the $S_i$ are stochastic maps controlled by the variable $A$.

The duotensor formalism of Ref.~\cite{hardy2011reformulating}, together with the above understanding of the geometry of GPT transformations, allow us to easily lift this result to arbitrary no-signalling channels in arbitrary tomographically local GPTs.
\begin{theorem}\label{thm:main}
Any non-signalling GPT channel $C$ in a tomographcically local GPT $\mathcal{G}$, can be written as an affine combination of product channels.
\end{theorem}
\proof
Consider some $n$-partite non-signalling channel $C$ in a tomographically local GPT:
\beq
\InputIfFileExists{Diagrams/GPTChannel.tikz}{}{\input{./figures/Diagrams/GPTChannel.tikz}}.
\eeq
By decomposing the input and output identities using Eq.~\eqref{eq:IdDecomp}, we obtain:
 \beq\label{eq:p1}
\InputIfFileExists{Diagrams/GPTChannel.tikz}{}{\input{./figures/Diagrams/GPTChannel.tikz}}\quad =\quad  %
\InputIfFileExists{Diagrams/GPTChannel3.tikz}{}{\input{./figures/Diagrams/GPTChannel3.tikz}}
\eeq
 We can then note that:
 \beq
\InputIfFileExists{Diagrams/Proof1.tikz}{}{\input{./figures/Diagrams/Proof1.tikz}}\quad :=\quad %
\InputIfFileExists{Diagrams/GPTChannel2.tikz}{}{\input{./figures/Diagrams/GPTChannel2.tikz}}
 \eeq
 is a non-signalling stochastic map, and hence we can apply the result of Ref.~\cite{al2013simulating} to obtain:
\beq \label{eq:DPDecomp}
\InputIfFileExists{Diagrams/Proof1.tikz}{}{\input{./figures/Diagrams/Proof1.tikz}}\quad  =\quad  \sum_{\alpha \in A} q_\alpha\  %
\InputIfFileExists{Diagrams/Proof2.tikz}{}{\input{./figures/Diagrams/Proof2.tikz}}
 \eeq 
 where $s_i^\alpha$ are stochastic maps, $q_\alpha \in \mathds{R}$, and $\sum_{\alpha \in A} q_\alpha =1$. By substituting this back in Eq.~\eqref{eq:p1}, we obtain:
\beq
\InputIfFileExists{Diagrams/GPTChannel.tikz}{}{\input{./figures/Diagrams/GPTChannel.tikz}}\quad  =\quad  \sum_{\alpha \in A} q_\alpha \ %
\InputIfFileExists{Diagrams/Proof3.tikz}{}{\input{./figures/Diagrams/Proof3.tikz}}\quad  =:\quad  \sum_{\alpha \in A} q_\alpha \ %
\InputIfFileExists{Diagrams/Proof4.tikz}{}{\input{./figures/Diagrams/Proof4.tikz}}
\eeq
where
\beq
\InputIfFileExists{Diagrams/Proof5.tikz}{}{\input{./figures/Diagrams/Proof5.tikz}}\quad  :=\quad  %
\InputIfFileExists{Diagrams/Proof6.tikz}{}{\input{./figures/Diagrams/Proof6.tikz}}.
\eeq
It is then easy to check that the $x_i^\alpha$ are discard preserving:
\begin{align}
\InputIfFileExists{Diagrams/Proof7.tikz}{}{\input{./figures/Diagrams/Proof7.tikz}}\quad  &=\quad %
\InputIfFileExists{Diagrams/Proof8.tikz}{}{\input{./figures/Diagrams/Proof8.tikz}}\quad =\quad %
\InputIfFileExists{Diagrams/Proof9.tikz}{}{\input{./figures/Diagrams/Proof9.tikz}}\quad\\ &=\quad %
\InputIfFileExists{Diagrams/Proof10.tikz}{}{\input{./figures/Diagrams/Proof10.tikz}}\quad =\quad %
\InputIfFileExists{Diagrams/Proof11.tikz}{}{\input{./figures/Diagrams/Proof11.tikz}}\quad =\quad %
\InputIfFileExists{Diagrams/Proof12.tikz}{}{\input{./figures/Diagrams/Proof12.tikz}}\quad\\ &=\quad %
\InputIfFileExists{Diagrams/Proof13.tikz}{}{\input{./figures/Diagrams/Proof13.tikz}}\ .
\end{align}
We can then use Thm.~\ref{theorem:AffSpanTransf} to write each $x_i^\alpha$ as an affine combination of GPT channels:
\beq
\InputIfFileExists{Diagrams/Proof5.tikz}{}{\input{./figures/Diagrams/Proof5.tikz}}\quad  =\quad \sum_{\beta} r_{i}^{\alpha\beta}\  %
\InputIfFileExists{Diagrams/NProof1.tikz}{}{\input{./figures/Diagrams/NProof1.tikz}}.
\eeq
We can write this instead as:
\beq
\InputIfFileExists{Diagrams/Proof5.tikz}{}{\input{./figures/Diagrams/Proof5.tikz}}\quad  =\quad %
\InputIfFileExists{Diagrams/extra1.tikz}{}{\input{./figures/Diagrams/extra1.tikz}}
\eeq
where ${C'}_i^\alpha$ is a classically controlled channel and $R_i^\alpha$ is a quasidistribution.

Now, let us define:
\beq
\InputIfFileExists{Diagrams/extra2.tikz}{}{\input{./figures/Diagrams/extra2.tikz}}
\eeq
such that
\beq
\InputIfFileExists{Diagrams/extra3.tikz}{}{\input{./figures/Diagrams/extra3.tikz}}\quad =\quad %
\InputIfFileExists{Diagrams/extra1.tikz}{}{\input{./figures/Diagrams/extra1.tikz}},
\eeq
where ${C'}_i$ is a classically controlled channel and $R_i$ is a quasistochastic map.

Putting this together with Eq.~\eqref{eq:DPDecomp} we find that:
\begin{align}
\InputIfFileExists{Diagrams/GPTChannel.tikz}{}{\input{./figures/Diagrams/GPTChannel.tikz}}\quad &=\quad 
 \sum_{\alpha} q_\alpha\  %
\InputIfFileExists{Diagrams/Proof4.tikz}{}{\input{./figures/Diagrams/Proof4.tikz}} \\
&=\quad \sum_\alpha q_\alpha %
\InputIfFileExists{Diagrams/extra4.tikz}{}{\input{./figures/Diagrams/extra4.tikz}} \cdots %
\InputIfFileExists{Diagrams/extra5.tikz}{}{\input{./figures/Diagrams/extra5.tikz}} \\
&=\quad %
\InputIfFileExists{Diagrams/extra6.tikz}{}{\input{./figures/Diagrams/extra6.tikz}},\label{eq:Qpp}
\end{align}
where $Q'$ is the quasidistrubtion defined by the $q_\alpha$, i.e.:
\beq
\InputIfFileExists{Diagrams/newQuasi3.tikz}{}{\input{./figures/Diagrams/newQuasi3.tikz}} \quad := \ q_\alpha
\eeq
for all $\alpha$.

Next, let us define $Q''$ by
\beq
\InputIfFileExists{Diagrams/newQuasi2.tikz}{}{\input{./figures/Diagrams/newQuasi2.tikz}}\quad:=\quad%
\InputIfFileExists{Diagrams/newQuasi1.tikz}{}{\input{./figures/Diagrams/newQuasi1.tikz}}
\eeq
such that we can now combine this with Eq.~\eqref{eq:Qpp} to write our channel as
\beq \label{eq:newlabel1}
\InputIfFileExists{Diagrams/GPTChannel.tikz}{}{\input{./figures/Diagrams/GPTChannel.tikz}}\quad =\quad %
\InputIfFileExists{Diagrams/extra7.tikz}{}{\input{./figures/Diagrams/extra7.tikz}}.
\eeq
Eq.~\ref{eq:newlabel1} gives us a quasiprobability distribution over a set of variables, one for each $C_i'$. In the remainder of this proof we show that this can be rewritten as a quasidistribution over a single variable, which is then copied to each of the $C_i'$'s. Diagrammatically, this means that the copy operation should be the last operation prior to the $C_i'$'s. It is then this quasidistribution over a single variable which defines our affine combination of product channels.

Now, define ``all but system $i$'' marginalisation maps, $D_i$ as:
\beq
\InputIfFileExists{Diagrams/extra8.tikz}{}{\input{./figures/Diagrams/extra8.tikz}}\quad =\quad %
\InputIfFileExists{Diagrams/extra9.tikz}{}{\input{./figures/Diagrams/extra9.tikz}},
\eeq
where the case $i\neq 1$ follows similarly.

We can then write:
\begin{align}
\InputIfFileExists{Diagrams/GPTChannel.tikz}{}{\input{./figures/Diagrams/GPTChannel.tikz}} 
\quad &=\quad  %
\InputIfFileExists{Diagrams/extra7.tikz}{}{\input{./figures/Diagrams/extra7.tikz}} \\
&=\quad  %
\InputIfFileExists{Diagrams/extra10.tikz}{}{\input{./figures/Diagrams/extra10.tikz}} \\
&=\quad  %
\InputIfFileExists{Diagrams/extra11.tikz}{}{\input{./figures/Diagrams/extra11.tikz}} \\
&=\quad  %
\InputIfFileExists{Diagrams/extra14.tikz}{}{\input{./figures/Diagrams/extra14.tikz}}\\
 &= \quad %
\InputIfFileExists{Diagrams/extra12.tikz}{}{\input{./figures/Diagrams/extra12.tikz}},
\end{align}
where in the last step we have simply merged together parallel wires into a single composite wire, whilst using Eq.~\eqref{eq:compCopy} to write the composite of copies as a copy of the composite.

By decomposing the quasidistribution $Q''$ we can equivalently write this as:
\beq
\InputIfFileExists{Diagrams/GPTChannel.tikz}{}{\input{./figures/Diagrams/GPTChannel.tikz}}\quad  =\quad  \sum_\beta {q_\beta''} \ %
\InputIfFileExists{Diagrams/extra13.tikz}{}{\input{./figures/Diagrams/extra13.tikz}},
\eeq
where $c_i^\beta$ are GPT channels and ${q_\beta''}$ is a quasidistribution. That is, any no-signalling GPT channel can be written as an affine combination of product GPT channels.
\endproof

Note that, if we have a GPT, such as quantum theory, in which one can always reversibly encode classical data into a GPT system, then we can rewrite this as:\begin{align}
\InputIfFileExists{Diagrams/GPTChannel.tikz}{}{\input{./figures/Diagrams/GPTChannel.tikz}}\quad  &=\quad  %
\InputIfFileExists{Diagrams/affineCominationGPT.tikz}{}{\input{./figures/Diagrams/affineCominationGPT.tikz}}\\
 &= \quad  %
\InputIfFileExists{Diagrams/affineCombinationGPT2.tikz}{}{\input{./figures/Diagrams/affineCombinationGPT2.tikz}} \\
 &= \quad  %
\InputIfFileExists{Diagrams/affineCombinationGPT3.tikz}{}{\input{./figures/Diagrams/affineCombinationGPT3.tikz}}
\end{align}
where $E$ is the encoding map, $D$ the decoding map, $s_Q$ is some vector which is not necessarily a physical GPT state, and where the $C'_i$ are GPT channels.

\section{Outlook}

In this work we have provided a characterisation of multipartite non-signalling channels in arbitrary locally-tomographic theories: these channels can always be represented as affine combinations of local channels. In the case where the input and output system types are classical, i.e., where the channel is a multipartite non-signalling stochastic map, we recover the result of Ref.~\cite{al2013simulating}. In the case of bipartite non-signalling channels whose inputs and outputs are quantum systems, we in turn recover the results of Ref.~\cite{chiribella2013quantum}. 

Our proof technique highlights the usefulness of the duotensor formalism, and we hope this will motivate its use throughout the quantum community. In particular, we show how it can be used to lift properties of multipartite stochastic maps, to arbitrary tomographically local GPTs. This motivates the question as to which other properties of stochastic maps can be similarly lifted?

\section*{Acknowledgements}

PJC, JHS, and ABS acknowledge support by the Foundation for Polish Science (IRAP project, ICTQT, contract no.2018/MAB/5, co-financed by EU within Smart Growth Operational Programme). 
All of the diagrams within this manuscript were prepared using TikZit.

\section*{Note added} 
We are not aware of this result appearing elsewhere in the literature, but suspect that this might be the case.  
We are happy to give credit if such a result is brought to our attention and will update the manuscript accordingly. 
We do expect, however, that our approach is novel and may find other applications in future works.

\bibliographystyle{unsrt}
\bibliography{bibliography}

\appendix

\section{A bit of convex geometry and a proof of Theorem~\ref{theorem:AffSpanTransf}}  
\label{AppBatman} 

This appendix aims to arrive at the proof of Theorem \ref{theorem:AffSpanTransf}. In order to do so, a background on convex geometry is provided, and the concepts presented are used to prove Lemmas \ref{lem:newlabel}-\ref{lem:keylemma}. Then, Lemmas \ref{lem:newlabel}, \ref{lem:intersection} and \ref{lem:keylemma} are directly used to prove Theorem \ref{theorem:AffSpanTransf}, while Lemma \ref{lastlemma} is used to prove Lemma \ref{lem:intersection}. All the sets $\DP$, $\PM$, and $\PMcone$ that are referred to here are defined in Section \ref{section:geometry}.

\subsection{Convex Geometry}
Given a real vector space $V$ and a (finite) set of vectors $v_1, \ldots, v_k \in V$, we define an \emph{affine combination} of those vectors to be of the form 
\beq 
\sum_{i=1}^k q_i v_i 
\eeq 
where $q_i \in \mathds{R}$ and $\sum_{i=1}^k q_i = 1$. 
Note the distinction between an affine combination and a convex combination, where the latter also requires that each $q_i$ is nonnegative.\footnote{This difference is analogous to the difference between quasiprobability distributions and (proper) probability distributions.}  
Given a set of vectors $C$, its \emph{affine hull} is the set of all affine combinations of vectors in $C$ and is denoted $\Aff(C)$. 
Lastly, if a set is equal to its affine hull, that is, it contains all of its affine combinations, then we say that the set is affine (or an affine space). 
Geometrically, one can view an affine space as a subspace translated by a fixed single vector. 

Conceptually, a convex set contains all the line segments between all pairs of points in the set. 
An affine set contains all the lines that extend beyond the endpoints of the line segments. 
This brings us to the definition of the \emph{core} of a set. 
Thinking of line segments, $x \in S$ is in the core of a set $S$ if for all $z \in V$, there exists a $t_z > 0$ such that $x + t z \in S$ for all $t \in [0, t_z]$.  
Conceptually, this means that given $x$, you can start drawing a line in \emph{any} direction and stay within the set $S$. 
Formally, 
\beq 
\core(S) := \{ x \in S | \forall z \in V, \exists t_z > 0, \text{ such that } x + t z \in S, \text{ for all } t \in [0, t_z] \}.  
\eeq 
  
\subsection{Lemmas and proof of Theorem \ref{theorem:AffSpanTransf}}
\begin{lemma} \label{lem:newlabel}
$\PM = \DP \cap \PMcone$. 
\end{lemma}

\begin{proof} 
Since $\PM \subseteq \DP \cap \PMcone$, all that remains to show is the opposite containment. 
Let $\phi \in \DP \cap \PMcone$ be a fixed, arbitrary vector. 
Since $\phi \in \PMcone$, we can write it as 
\beq 
\InputIfFileExists{Diagrams/K.tikz}{}{\input{./figures/Diagrams/K.tikz}}\quad =\quad \sum_{i=1}^k
\InputIfFileExists{Diagrams/effectstate4.tikz}{}{\input{./figures/Diagrams/effectstate4.tikz}}
\eeq  
where 
$k$ is finite, 
$s_1, \ldots, s_k\in \Omega_T$ are normalized states, and 
$e_1, \ldots, e_k\in \mathcal{T}_S$ are in the effect cone.

It remains to show that $e_1, \ldots, e_k$ sum to $\discard$. 
Since $\phi \in \DP$, we have that 
\beq
\InputIfFileExists{Diagrams/MPDiscard3.tikz}{}{\input{./figures/Diagrams/MPDiscard3.tikz}}\quad =\quad%
\InputIfFileExists{Diagrams/causalphi.tikz}{}{\input{./figures/Diagrams/causalphi.tikz}}\quad =\quad \sum_{i=1}^k %
\InputIfFileExists{Diagrams/MPDiscard.tikz}{}{\input{./figures/Diagrams/MPDiscard.tikz}}\quad =\quad \sum_{i=1}^k %
\InputIfFileExists{Diagrams/MPDiscard2.tikz}{}{\input{./figures/Diagrams/MPDiscard2.tikz}}.
\eeq
	since the states $s_i$ are normalised. This is the desired equality we seek.  Finally, note that by the partial order given in Section \ref{section:effects}, for all $ e_{i}$ in the sum above, we have $e_{i} \in \mathcal{E}_S$, so $ \phi$ satisfies all requirements for membership in $\PM$.
\end{proof}

\begin{lemma} \label{lastlemma}
Suppose $S$ is a set and 
\begin{equation} \label{otherKcone}
K = \left\{ \sum_{i=1}^n \alpha_i s_i \middle| \alpha_i \geq 0, s_i \in S, n \text{ finite} \right\}  
\end{equation} 
is a full-dimensional cone, i.e., $V = K - K$.  
	Suppose for $x \in K$, we have that for all $s \in S$, there exists $t_{s} > 0$ such that $x -  t s \in K$ for all $t \in [0, t_{s}]$. 
Then $x \in \core(K)$. 
\end{lemma} 

Note that the only difference between the definition of $\core(K)$ and the condition above is that the vectors in the statement above are not arbitrary but rather belong to a set which generates a full-dimensional cone. 

\begin{proof}[Proof of Lemma~\ref{lastlemma}]
	 Since $V = K-K$, for any arbitrary $v \in V$ we can write $v = y - z$ where $y, z \in K$.  
Then for $x \in K$, $t \geq 0$, we can write the following  
\begin{equation}
x + t v = x + t y - t z. 
\end{equation} 
	 So, for $x$ to be in $\core(K)$, it suffices to find a $t_{v} > 0$ such that $x + t v \in K$ for all $t \in [0, t_{v}]$. To do that, we have two cases to analyse, $z=0$ and $z \neq 0$.
	Note that if $z = 0$, $x + tv = x + t y \in K$ for all $t \geq 0$ since $x, y \in K$,  and $K$ is cone, so this case is trivial. 
Suppose $z \in K$ is nonzero, then we can write it as $\sum_i \alpha_i s_i$ where $\alpha_i > 0$ and $s_i \in S$ and the sum is finite. 
Then we have  
\begin{equation}
x + t v = x + t y - \sum_i \alpha_i t s_i. 
\end{equation} 
For brevity, define $a = \sum_i \alpha_i > 0$. 
By hypothesis, let $t_i > 0$ be such that 
\begin{equation} 
x - a t s_i \in K 
\end{equation} 
	 for all $t \in [0, t_i]$.
	This exists since $a$ is positive  and by assumption there is some $t_{i} ' = a t_{i}$ such that $x-t s_{i} \in K$ for all $t \in [0,t_{i} ']$. 

We now have 
\begin{align}
x + t v 
& = x + ty - tz \\ 
& = t y + x - \sum_i \alpha_i t s_i \\ 
& = t y + \frac{1}{a} \left( \sum_i \alpha_i x - \sum_i a \cdot \alpha_i t s_i \right) \\ 
& = t y + \frac{1}{a} \sum_i \alpha_i (x - a t s_i) 
\end{align} 
	which is in $K$ for all  $t \in [0, t_{v}]$ where $ t_{v} := \min_i \{ t_i \}$ (which is positive since there are finitely many indices $i$). 
This concludes the proof.  
\end{proof}

We use this particular case characterization of a core element to prove the following lemma which is helpful in our proof of Theorem~\ref{theorem:AffSpanTransf}.

\begin{lemma}\label{lem:intersection}
Let $\mu \in \intr(\mathcal{T}^T)$ be a normalised state.  
Then 
\beq
\InputIfFileExists{Diagrams/discardPrepare.tikz}{}{\input{./figures/Diagrams/discardPrepare.tikz}} \in \DP \cap \core(\PMcone), 
\eeq  
where, recall, $\discard \in \intr(E_S)$ is the discarding effect. 
\end{lemma} 

\begin{proof} 
Clearly 
\beq
\InputIfFileExists{Diagrams/discardPrepare.tikz}{}{\input{./figures/Diagrams/discardPrepare.tikz}} \in \DP,
\eeq
as $\mu$ is normalised,
 so all that remains to show is that it is in 
 $\core(\PMcone)$. 

Define 
\begin{equation} 
S := \left\{ %
\InputIfFileExists{Diagrams/effectstate5.tikz}{}{\input{./figures/Diagrams/effectstate5.tikz}}\ \middle| \ s \in \mathcal{T}^T, e \in \mathcal{T}_S \right\}. 
\end{equation} 
	Since $\PMcone$ (as defined in Equation~(\ref{Kcone})) is the convex hull of $S$ (as defined in Equation~(\ref{otherKcone})) and $\LST = \PMcone - \PMcone$, by Lemma~\ref{lastlemma}, it suffices to show that for a fixed $s \in \mathcal{T}^T$ and $e \in \mathcal{T}_S$, there exists $\hat{t} > 0$ such that  
\beq
\InputIfFileExists{Diagrams/discardPrepare.tikz}{}{\input{./figures/Diagrams/discardPrepare.tikz}} - t %
\InputIfFileExists{Diagrams/effectstate5.tikz}{}{\input{./figures/Diagrams/effectstate5.tikz}} \in \PMcone
\eeq for all $t \in [0, \hat{t}]$.  
For $t > 0$, we can write  
\begin{align} 
\InputIfFileExists{Diagrams/discardPrepare.tikz}{}{\input{./figures/Diagrams/discardPrepare.tikz}} - t %
\InputIfFileExists{Diagrams/effectstate5.tikz}{}{\input{./figures/Diagrams/effectstate5.tikz}} 
\quad &= \quad %
\InputIfFileExists{Diagrams/discardPrepare.tikz}{}{\input{./figures/Diagrams/discardPrepare.tikz}}  - \sqrt{t} %
\InputIfFileExists{Diagrams/su.tikz}{}{\input{./figures/Diagrams/su.tikz}}  + \sqrt{t} %
\InputIfFileExists{Diagrams/su.tikz}{}{\input{./figures/Diagrams/su.tikz}} - t %
\InputIfFileExists{Diagrams/effectstate5.tikz}{}{\input{./figures/Diagrams/effectstate5.tikz}}  \\ 
& =\quad \left( %
\InputIfFileExists{Diagrams/mu.tikz}{}{\input{./figures/Diagrams/mu.tikz}} - \sqrt{t} %
\InputIfFileExists{Diagrams/s.tikz}{}{\input{./figures/Diagrams/s.tikz}} \right)  %
\InputIfFileExists{Diagrams/trace.tikz}{}{\input{./figures/Diagrams/trace.tikz}} +  \sqrt{t} %
\InputIfFileExists{Diagrams/s.tikz}{}{\input{./figures/Diagrams/s.tikz}}  \left( %
\InputIfFileExists{Diagrams/trace.tikz}{}{\input{./figures/Diagrams/trace.tikz}} - \sqrt{t} %
\InputIfFileExists{Diagrams/e.tikz}{}{\input{./figures/Diagrams/e.tikz}}\right).    
\end{align} 
Note that 
\beq
\InputIfFileExists{Diagrams/mu.tikz}{}{\input{./figures/Diagrams/mu.tikz}} - \sqrt{t} %
\InputIfFileExists{Diagrams/s.tikz}{}{\input{./figures/Diagrams/s.tikz}} \ \ \in \ \mathcal{T}^T
\qquad \text{and}\qquad
\InputIfFileExists{Diagrams/trace.tikz}{}{\input{./figures/Diagrams/trace.tikz}} - \sqrt{t} %
\InputIfFileExists{Diagrams/e.tikz}{}{\input{./figures/Diagrams/e.tikz}}\ \ \in \ \mathcal{T}_S
\eeq
for all sufficiently small $t > 0$  as $\mu$ and $\discard$ are interior in their respective cones.  
Therefore, 
\beq
\InputIfFileExists{Diagrams/discardPrepare.tikz}{}{\input{./figures/Diagrams/discardPrepare.tikz}} - t %
\InputIfFileExists{Diagrams/effectstate5.tikz}{}{\input{./figures/Diagrams/effectstate5.tikz}} \in \PMcone
\eeq
 is in $\PMcone$ for all $t > 0$ sufficiently small.
This concludes the proof. 
\end{proof}

\begin{lemma} \label{lem:keylemma} 
Given a real vector space $V$, let $S \subseteq V$ be a subset and let $A \subseteq V$ be an affine space. 
If $\core(S) \cap A \neq \emptyset$, then $\Aff(S \cap A) = A$. 
\end{lemma} 

Before diving into the proof, we explain the idea first since the proof is actually quite simple to picture geometrically, but the proof we give here is algebraic. 
We start with a point $x \in \core(S) \cap A$ and draw the line segment between that point to some other arbitrary fixed point $z \in A$. 
Since $x \in \core(S)$, there is a point on that line segment, call it $y$, such that it is still in $S$. 
And since $A$ is affine, $y$ is in $A$ as well. 
The proof concludes by noting that $z$ is an affine combination of $x$ and $y$. 
\beq
\includegraphics[clip, trim=1cm 13.5cm 3cm 8cm,scale=.7]{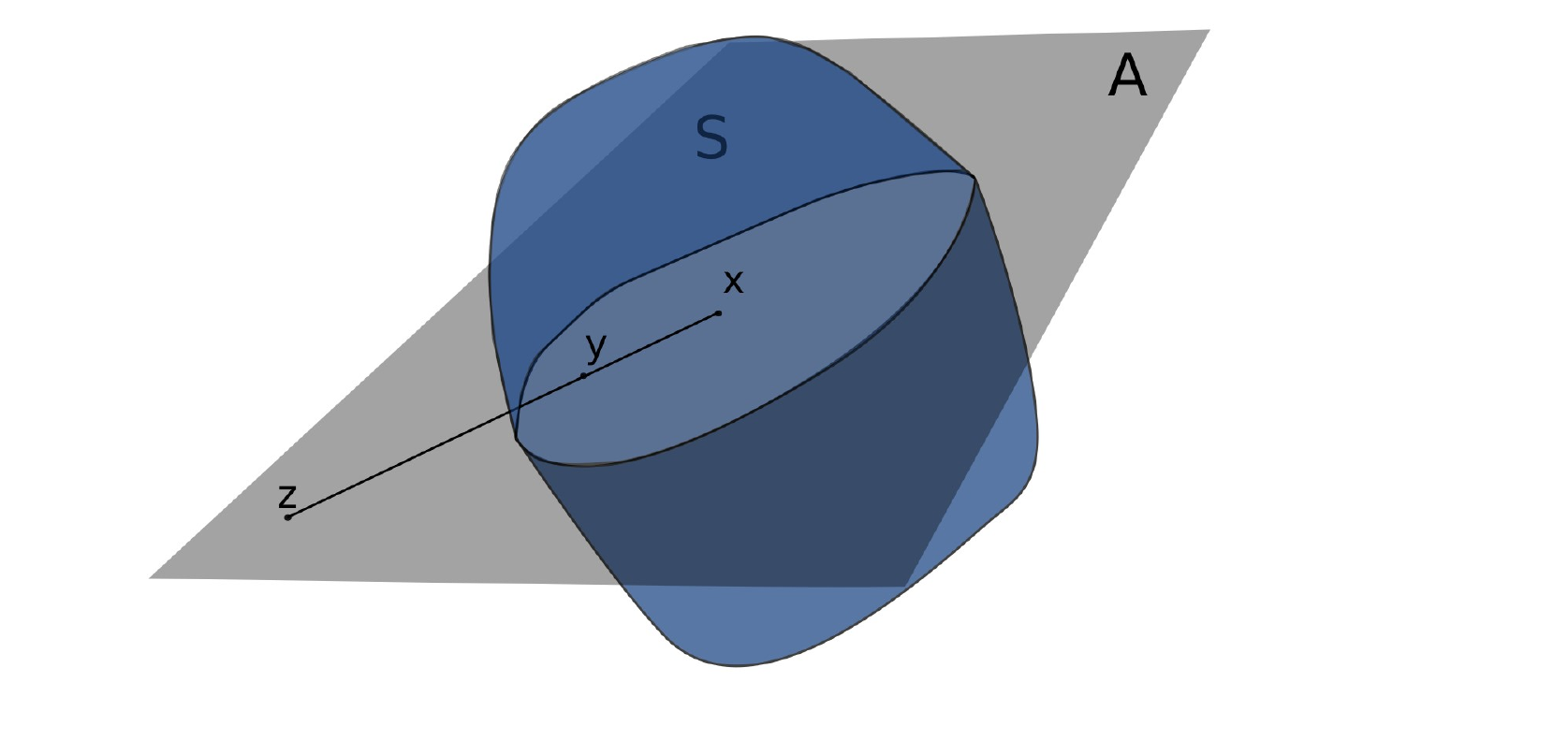}
\eeq

\begin{proof}[Proof of Lemma~\ref{lem:keylemma}]
Since $S \cap A \subseteq A$, we have that $\Aff(S \cap A) \subseteq \Aff(A) = A$. 
Therefore, all that remains to show is the reverse containment. 
To this end, let $z \in A$ be a fixed, arbitrary vector and let 
$x \in \core(S) \cap A$ (which exists by hypothesis). 
Define 
\beq 
y = (1-t) x + t z, 
\eeq 
for some $t > 0$ which we define momentarily.  
Notice that $y$ is an affine combination of $x$ and $z$, both of which are in $A$, and thus $y \in A$ as well. 
We now want to show that $y \in S$. 
Note that $y$ can be rewritten as 
\beq 
y = x + t (z-x).
\eeq 
Since $x \in \core(S)$, there exists a ${t} \in (0,1)$ such that $y \in S$. 
Thus, we have $x$ and $y$ both belonging to $S \cap A$.  
Notice that 
\beq
z = \left( \frac{1}{{t}} \right) y + \left( \frac{{t}-1}{{t}} \right) x. 
\eeq 
Since $z$ is an affine combination of $x$ and $y$, both belonging to $S \cap A$, the result follows. 
\end{proof} 

Lemmas 1-4 together allow us to write a short proof for Theorem \ref{theorem:AffSpanTransf}. 

\begin{proof}[\textbf{Proof of Theorem~\ref{theorem:AffSpanTransf}}]
	We want to show that $\DP = \Aff(\PM)$.

	By Lemma \ref{lem:newlabel}, we know
	\begin{equation}\label{}
		\PM = \DP \cap \PMcone.
	\end{equation}
	Now, Lemma \ref{lem:intersection} tells us that 
	\begin{equation}\label{}
		\DP \cap \core(\PMcone) \neq \emptyset. 
	\end{equation}
	If we set $A:=\DP$ and $S:=\PMcone$, the assumptions in Lemma \ref{lem:keylemma} are satisfied and we can use it to conclude that
	\begin{equation}\label{}
		\DP = \Aff( \PMcone \cap \DP ) = \Aff( \DP \cap \PMcone) = \Aff( \PM).
	\end{equation}
	
\end{proof}

\end{document}

%% file: preamble.tex
\usepackage{rotating}
\usepackage{floatpag}
\rotfloatpagestyle{empty}
\renewcommand{\theenumi}{\arabic{enumi}}
\usepackage{xcolor}
\definecolor{googleblue}{RGB}{34, 0, 204}

\usepackage[linktocpage=true, colorlinks=true, linkcolor=googleblue, urlcolor=googleblue, citecolor=blue]{hyperref}

\usepackage{dsfont}
\usepackage{amsmath,amsthm,amssymb}
\usepackage{xspace,enumerate,epsfig}
\usepackage{graphicx}
\graphicspath{{.}{./figures/}}
\usepackage{marvosym}

\usepackage{tikzfig}
\usepackage{stmaryrd}
\usepackage{docmute}
\usepackage{keycommand}

\input{defs.tex}

\input{tikzstyles.tex}

\let\olddagger\dagger
\renewcommand{\dagger}{\ensuremath{\olddagger}\xspace}

\theoremstyle{definition}
\newtheorem{theorem}{Theorem}[section]

\newtheorem{lemma}[]{Lemma}

\newtheorem{example*}[theorem]{Example*}
\newtheorem{examples*}[theorem]{Examples*}

\newtheorem{remark*}[theorem]{Remark*}

\theoremstyle{plain}

\usepackage{color}
\def\bR{\begin{color}{red}}
\def\bB{\begin{color}{blue}}
\def\bM{\begin{color}{magenta}}
\def\bC{\begin{color}{cyan}}
\def\bW{\begin{color}{white}}
\def\bBl{\begin{color}{black}}
\def\bG{\begin{color}{green}}
\def\bY{\begin{color}{yellow}}
\def\e{\end{color}\xspace}
\newcommand{\bit}{\begin{itemize}}
\newcommand{\eit}{\end{itemize}\par\noindent}
\newcommand{\ben}{\begin{enumerate}}
\newcommand{\een}{\end{enumerate}\par\noindent}
\newcommand{\beq}{\begin{equation}}
\newcommand{\eeq}{\end{equation}\par\noindent}
\newcommand{\beqa}{\begin{eqnarray*}}
\newcommand{\eeqa}{\end{eqnarray*}\par\noindent}
\newcommand{\beqn}{\begin{eqnarray}}
\newcommand{\eeqn}{\end{eqnarray}\par\noindent}

\def\jR{\begin{color}{DarkRed}}
\def\jB{\begin{color}{MidnightBlue}}
\def\jM{\begin{color}{magenta}}
\def\jC{\begin{color}{cyan}}
\def\jW{\begin{color}{white}}
\def\jBl{\begin{color}{black}}
\def\jG{\begin{color}{green}}
\def\jY{\begin{color}{yellow}}

\def\cR{\begin{color}{Crimson}}
\def\cB{\begin{color}{cyan}}
\def\cM{\begin{color}{magenta}}
\def\cC{\begin{color}{cyan}}
\def\cW{\begin{color}{white}}
\def\cBl{\begin{color}{black}}
\def\cG{\begin{color}{green}}
\def\cY{\begin{color}{yellow}}


%% file: defs.tex
\newkeycommand{\pointmap}[style=point][1]{\,\tikz{\node[style=\commandkey{style}] (x) at (0,-0.3) {$#1$};\node [style=none] (3) at (0, 1.0) {};\node [style=none] (3) at (0, -1.0) {};\draw (x)--(0,0.7);}\,}

\newkeycommand{\typedkpoint}[style=kpoint,edgestyle=][2]{%
\,\begin{tikzpicture}
  \begin{pgfonlayer}{nodelayer}
    \node [style=none] (0) at (0, 1) {};
    \node [style=\commandkey{style}] (1) at (0, -0.75) {$#2$};
    \node [style=right label] (2) at (0.25, 0.5) {$#1$};
  \end{pgfonlayer}
  \begin{pgfonlayer}{edgelayer}
    \draw [\commandkey{edgestyle}] (1) to (0.center);
  \end{pgfonlayer}
\end{tikzpicture}\,}

\newkeycommand{\typedkpointdag}[style=kpoint adjoint,edgestyle=][2]{%
\,\begin{tikzpicture}
  \begin{pgfonlayer}{nodelayer}
    \node [style=none] (0) at (0, -1) {};
    \node [style=\commandkey{style}] (1) at (0, 0.75) {$#2$};
    \node [style=right label] (2) at (0.25, -0.5) {$#1$};
  \end{pgfonlayer}
  \begin{pgfonlayer}{edgelayer}
    \draw [\commandkey{edgestyle}] (0.center) to (1);
  \end{pgfonlayer}
\end{tikzpicture}\,}

\newcommand{\bistateadj}[1]{\,\begin{tikzpicture}[yshift=-3mm]
  \begin{pgfonlayer}{nodelayer}
    \node [style=kpointadj, minimum width=1 cm, inner sep=2pt] (0) at (0, 0.5) {$\psi$};
    \node [style=none] (1) at (-0.75, 0.25) {};
    \node [style=none] (2) at (0.75, 0.25) {};
    \node [style=none] (3) at (-0.75, -0.5) {};
    \node [style=none] (4) at (0.75, -0.5) {};
  \end{pgfonlayer}
  \begin{pgfonlayer}{edgelayer}
    \draw (1.center) to (3.center);
    \draw (2.center) to (4.center);
  \end{pgfonlayer}
\end{tikzpicture}\,}

\newkeycommand{\pointketbra}[style1=copoint,style2=point][2]{\,%
\begin{tikzpicture}
  \begin{pgfonlayer}{nodelayer}
    \node [style=none] (0) at (0, -1.5) {};
    \node [style=\commandkey{style1}] (1) at (0, -0.75) {$#1$};
    \node [style=\commandkey{style2}] (2) at (0, 0.75) {$#2$};
    \node [style=none] (3) at (0, 1.5) {};
  \end{pgfonlayer}
  \begin{pgfonlayer}{edgelayer}
    \draw (0.center) to (1);
    \draw (2) to (3.center);
  \end{pgfonlayer}
\end{tikzpicture}\,}

\newkeycommand{\twocopointketbra}[style1=copoint,style2=copoint,style3=point][3]{\,%
\begin{tikzpicture}
  \begin{pgfonlayer}{nodelayer}
    \node [style=none] (0) at (-0.75, -1.5) {};
    \node [style=none] (1) at (0.75, -1.5) {};
    \node [style=\commandkey{style1}] (2) at (-0.75, -0.75) {$#1$};
    \node [style=\commandkey{style2}] (3) at (0.75, -0.75) {$#2$};
    \node [style=\commandkey{style3}] (4) at (0, 0.75) {$#3$};
    \node [style=none] (5) at (0, 1.5) {};
  \end{pgfonlayer}
  \begin{pgfonlayer}{edgelayer}
    \draw (0.center) to (2);
    \draw (1.center) to (3);
    \draw (4) to (5.center);
  \end{pgfonlayer}
\end{tikzpicture}\,}

\newkeycommand{\twopointketbra}[style1=copoint,style2=point,style3=point][3]{\,%
\begin{tikzpicture}
  \begin{pgfonlayer}{nodelayer}
    \node [style=none] (0) at (0, -1.5) {};
    \node [style=none] (1) at (-0.75, 1.5) {};
    \node [style=\commandkey{style1}] (2) at (0, -0.75) {$#1$};
    \node [style=\commandkey{style2}] (3) at (-0.75, 0.75) {$#2$};
    \node [style=\commandkey{style3}] (4) at (0.75, 0.75) {$#3$};
    \node [style=none] (5) at (0.75, 1.5) {};
  \end{pgfonlayer}
  \begin{pgfonlayer}{edgelayer}
    \draw (0.center) to (2);
    \draw (1.center) to (3);
    \draw (4) to (5.center);
  \end{pgfonlayer}
\end{tikzpicture}\,}

\newkeycommand{\pointbraket}[style1=point,style2=copoint][2]{\,%
\begin{tikzpicture}
  \begin{pgfonlayer}{nodelayer}
    \node [style=\commandkey{style1}] (0) at (0, -0.5) {$#1$};
    \node [style=\commandkey{style2}] (1) at (0, 0.5) {$#2$};
  \end{pgfonlayer}
  \begin{pgfonlayer}{edgelayer}
    \draw (0) to (1);
  \end{pgfonlayer}
\end{tikzpicture}\,}

\newkeycommand{\kpointbraket}[style1=kpoint,style2=kpoint adjoint][2]{\,%
\begin{tikzpicture}
  \begin{pgfonlayer}{nodelayer}
    \node [style=\commandkey{style1}] (0) at (0, -0.75) {$#1$};
    \node [style=\commandkey{style2}] (1) at (0, 0.75) {$#2$};
  \end{pgfonlayer}
  \begin{pgfonlayer}{edgelayer}
    \draw (0) to (1);
  \end{pgfonlayer}
\end{tikzpicture}\,}

\newkeycommand{\kpointmap}[style1=kpoint,style2=map][2]{\,%
\begin{tikzpicture}
  \begin{pgfonlayer}{nodelayer}
    \node [style=\commandkey{style1}] (0) at (0, -1.1) {$#1$};
    \node [style=\commandkey{style2}] (1) at (0, 0.2) {$#2$};
    \node [style=none] (2) at (0, 1.5) {};
  \end{pgfonlayer}
  \begin{pgfonlayer}{edgelayer}
    \draw (0) to (1);
    \draw (1) to (2.center);
  \end{pgfonlayer}
\end{tikzpicture}%
\,}

\newkeycommand{\sandwichtwo}[style1=point,style2=map,style3=map,style4=copoint][4]{\,%
\begin{tikzpicture}
  \begin{pgfonlayer}{nodelayer}
    \node [style=\commandkey{style2}] (0) at (0, -0.75) {$#2$};
    \node [style=\commandkey{style3}] (1) at (0, 0.75) {$#3$};
    \node [style=\commandkey{style1}] (2) at (0, -2) {$#1$};
    \node [style=\commandkey{style4}] (3) at (0, 2) {$#4$};
  \end{pgfonlayer}
  \begin{pgfonlayer}{edgelayer}
    \draw (2) to (0);
    \draw (0) to (1);
    \draw (1) to (3);
  \end{pgfonlayer}
\end{tikzpicture}\,}

\newkeycommand{\longbraket}[style1=point,style2=copoint][2]{\,%
\begin{tikzpicture}
  \begin{pgfonlayer}{nodelayer}
    \node [style=\commandkey{style1}] (0) at (0, -2) {$#1$};
    \node [style=\commandkey{style2}] (1) at (0, 2) {$#2$};
  \end{pgfonlayer}
  \begin{pgfonlayer}{edgelayer}
    \draw (0) to (1);
  \end{pgfonlayer}
\end{tikzpicture}\,}

\newkeycommand{\onb}[style=point]{\ensuremath{\left\{\tikz{\node[style=\commandkey{style}] (x) at (0,-0.3) {$j$};\draw (x)--(0,0.7);}\right\}}\xspace}

\newkeycommand{\copointmap}[style=copoint][1]{\,\tikz{\node[style=\commandkey{style}] (x) at (0,0.3) {$#1$};\draw (x)--(0,-0.7);}\,}

\tikzstyle{cdiag}=[matrix of math nodes, row sep=3em, column sep=3em, text height=1.5ex, text depth=0.25ex,inner sep=0.5em]
\tikzstyle{arrow above}=[transform canvas={yshift=0.5ex}]
\tikzstyle{arrow below}=[transform canvas={yshift=-0.5ex}]



%% file: tikzstyles.tex
\usetikzlibrary{shapes.multipart}

\tikzstyle{infpoint}=[regular polygon,regular polygon sides=3,draw,scale=0.75,inner sep=-0.5pt,minimum width=9mm,fill=white,regular polygon rotate=90]
\tikzstyle{infcopoint}=[regular polygon,regular polygon sides=3,draw,scale=0.75,inner sep=-0.5pt,minimum width=9mm,fill=white,regular polygon rotate=270]
\tikzstyle{fMeas}=[regular polygon,regular polygon sides=3,draw,scale=0.75,inner sep=-0.5pt,minimum width=4mm,fill=white,regular polygon rotate=0,line width=1pt]
\tikzstyle{fPrep}=[regular polygon,regular polygon sides=3,draw,scale=0.75,inner sep=-0.5pt,minimum width=4mm,fill=white,regular polygon rotate=180,line width=1pt]
\tikzstyle{fTransition}=[fill=white,draw, line width = 1pt,inner sep=0.6mm,font=\footnotesize,minimum height=2mm,minimum width=2mm]
 \tikzstyle{rightground}=[circuit ee IEC,thick,ground,rotate=0,xscale=2.5,yscale=2]

\tikzstyle{bwSpider}=[
       rectangle split,
       rectangle split parts=2,
       rectangle split part fill={black,white},
 minimum size=3.6 mm, inner sep=-2mm, draw=black,scale=0.5,rounded corners=0.8 mm
       ]
 \tikzstyle{wbSpider}=[
       rectangle split,
       rectangle split parts=2,
       rectangle split part fill={white,black},
 minimum size=3.6 mm, inner sep=-2mm, draw=black,scale=0.5,rounded corners=0.8 mm
       ]
\tikzstyle{qWire}=[line width = 1pt, color=quantumviolet]
\tikzstyle{pWire}=[line width = 1pt, color=red!60!black]
\tikzstyle{gWire}=[line width = 1.5pt, color=green!60!black]
\tikzstyle{cWire}=[color=quantumgray,thin]
\tikzstyle{env}=[copoint,regular polygon rotate=0,minimum width=0.2cm, fill=black]

\tikzstyle{probs}=[shape=semicircle,fill=white,draw=black,shape border rotate=180,minimum width=1.2cm]

\tikzstyle{every picture}=[baseline=-0.25em,scale=0.5]
\tikzstyle{dotpic}=[] 
\tikzstyle{diredges}=[every to/.style={diredge}]
\tikzstyle{math matrix}=[matrix of math nodes,left delimiter=(,right delimiter=),inner sep=2pt,column sep=1em,row sep=0.5em,nodes={inner sep=0pt},text height=1.5ex, text depth=0.25ex]

\tikzstyle{inline text}=[text height=1.5ex, text depth=0.25ex,yshift=0.5mm]
\tikzstyle{label}=[font=\footnotesize,text height=1.5ex, text depth=0.25ex,yshift=0.5mm]
\tikzstyle{left label}=[label,anchor=east,xshift=1.5mm]
\tikzstyle{right label}=[label,anchor=west,xshift=-1mm]
\tikzstyle{up label}=[label,anchor=south,yshift=-1mm]

\tikzstyle{braceedge}=[decorate,decoration={brace,amplitude=2mm,raise=-1mm}]
\tikzstyle{small braceedge}=[decorate,decoration={brace,amplitude=1mm,raise=-1mm}]

\tikzstyle{doubled}=[line width=1.6pt] 
\tikzstyle{boldedge}=[doubled,shorten <=-0.17mm,shorten >=-0.17mm]
\tikzstyle{boldedgegray}=[doubled,gray,shorten <=-0.17mm,shorten >=-0.17mm]
\tikzstyle{singleedgegray}=[gray]

\tikzstyle{semidoubled}=[line width=1.4pt] 
\tikzstyle{semiboldedgegray}=[semidoubled,gray,shorten <=-0.17mm,shorten >=-0.17mm]

\tikzstyle{boxedge}=[semiboldedgegray]

\tikzstyle{boldedgedashed}=[very thick,dashed,shorten <=-0.17mm,shorten >=-0.17mm]
\tikzstyle{vboldedgedashed}=[doubled,dashed,shorten <=-0.17mm,shorten >=-0.17mm]
\tikzstyle{left hook arrow}=[left hook-latex]
\tikzstyle{right hook arrow}=[right hook-latex]
\tikzstyle{sembracket}=[line width=0.5pt,shorten <=-0.07mm,shorten >=-0.07mm]

\tikzstyle{causal edge}=[->,thick,gray]
\tikzstyle{causal nondir}=[thick,gray]
\tikzstyle{timeline}=[thick,gray, dashed]

\tikzstyle{cedge}=[<->,thick,gray!70!white]

\tikzstyle{empty diagram}=[draw=gray!40!white,dashed,shape=rectangle,minimum width=1cm,minimum height=1cm]
\tikzstyle{empty diagram small}=[draw=gray!50!white,dashed,shape=rectangle,minimum width=0.6cm,minimum height=0.5cm]

\tikzstyle{dot}=[inner sep=0mm,minimum width=2mm,minimum height=2mm,draw,shape=circle]
\tikzstyle{leak}=[white dot, shape=regular polygon, minimum size=3.3 mm, regular polygon sides=3, outer sep=-0.2mm, regular polygon rotate=270]
\tikzstyle{proj}=[draw,fill=white,chamfered rectangle,chamfered rectangle angle=30, minimum width=2mm, minimum height= 1mm,scale=0.5, outer sep=-0.2mm]
\tikzstyle{wide proj}=[draw,fill=white,chamfered rectangle,chamfered rectangle angle=30, minimum width=15mm,minimum height=1mm,scale=0.5, outer sep=-0.2mm]
\tikzstyle{very wide proj}=[draw,fill=white,chamfered rectangle,chamfered rectangle angle=30, minimum width=25mm,minimum height=1mm,scale=0.5, outer sep=-0.2mm]
\tikzstyle{very very wide proj}=[draw,fill=white,chamfered rectangle,chamfered rectangle angle=30, minimum width=35mm,minimum height=1mm,scale=0.5, outer sep=-0.2mm]
\tikzstyle{preleak}=[proj]

\tikzstyle{split proj out}=[regular polygon,regular polygon sides=3,draw,scale=0.75,inner sep=-0.5pt,minimum width=3.3mm,fill=white,regular polygon rotate=180]
\tikzstyle{split proj in}=[regular polygon,regular polygon sides=3,draw,scale=0.75,inner sep=-0.5pt,minimum width=3.3mm,fill=white]
\tikzstyle{Vleak}=[white dot, shape=regular polygon, minimum size=3.3 mm, regular polygon sides=3, outer sep=-0.2mm, regular polygon rotate=90]
\tikzstyle{dleak}=[white dot, line width=1.6pt, shape=regular polygon, minimum size=3.3 mm, regular polygon sides=3, outer sep=-0.2mm, regular polygon rotate=270]

\tikzstyle{Wsquare}=[white dot, shape=regular polygon, rounded corners=0.8 mm, minimum size=3.3 mm, regular polygon sides=3, outer sep=-0.2mm]
\tikzstyle{Wsquareadj}=[white dot, shape=regular polygon, rounded corners=0.8 mm, minimum size=3.3 mm, regular polygon sides=3, outer sep=-0.2mm, regular polygon rotate=180]
\tikzstyle{ddot}=[inner sep=0mm, doubled, minimum width=2.5mm,minimum height=2.5mm,draw,shape=circle]

\tikzstyle{black dot}=[dot,fill=black]
\tikzstyle{white dot}=[dot,fill=white,,text depth=-0.2mm]
\tikzstyle{white Wsquare}=[Wsquare,fill=gray,,text depth=-0.2mm]
\tikzstyle{white Wsquareadj}=[Wsquareadj,fill=white,,text depth=-0.2mm]
\tikzstyle{green dot}=[white dot] 
\tikzstyle{gray dot}=[dot,fill=gray!40!white,,text depth=-0.2mm]
\tikzstyle{red dot}=[gray dot] 

\tikzstyle{black ddot}=[ddot,fill=black]
\tikzstyle{white ddot}=[ddot,fill=white]
\tikzstyle{gray ddot}=[ddot,fill=gray!40!white]

\tikzstyle{gray edge}=[gray!60!white]

\tikzstyle{small dot}=[inner sep=0.5mm,minimum width=0pt,minimum height=0pt,draw,shape=circle]

\tikzstyle{small black dot}=[small dot,fill=black]
\tikzstyle{small white dot}=[small dot,fill=white]
\tikzstyle{small gray dot}=[small dot,fill=gray!40!white]

\tikzstyle{causal dot}=[inner sep=0.4mm,minimum width=0pt,minimum height=0pt,draw=white,shape=circle,fill=gray!40!white]

\tikzstyle{phase dimensions}=[minimum size=5mm,font=\footnotesize,rectangle,rounded corners=2.5mm,inner sep=0.2mm,outer sep=-2mm]
\tikzstyle{dphase dimensions}=[minimum size=5mm,font=\footnotesize,rectangle,rounded corners=2.5mm,inner sep=0.2mm,outer sep=-2mm]

\tikzstyle{white phase dot}=[dot,fill=white,phase dimensions]
\tikzstyle{white phase ddot}=[ddot,fill=white,dphase dimensions]

\tikzstyle{white rect ddot}=[draw=black,fill=white,doubled,minimum size=5mm,font=\footnotesize,rectangle,rounded corners=2.5mm,inner sep=0.2mm]
\tikzstyle{gray rect ddot}=[draw=black,fill=gray!40!white,doubled,minimum size=6mm,font=\footnotesize,rectangle,rounded corners=3mm]

\tikzstyle{gray phase dot}=[dot,fill=gray!40!white,phase dimensions]
\tikzstyle{gray phase ddot}=[ddot,fill=gray!40!white,dphase dimensions]
\tikzstyle{grey phase dot}=[gray phase dot]
\tikzstyle{grey phase ddot}=[gray phase ddot]

\tikzstyle{small phase dimensions}=[minimum size=4mm,font=\tiny,rectangle,rounded corners=2mm,inner sep=0.2mm,outer sep=-2mm]
\tikzstyle{small dphase dimensions}=[minimum size=4mm,font=\tiny,rectangle,rounded corners=2mm,inner sep=0.2mm,outer sep=-2mm]

\tikzstyle{small gray phase dot}=[dot,fill=gray!40!white,small phase dimensions]
\tikzstyle{small gray phase ddot}=[ddot,fill=gray!40!white,small dphase dimensions]

\tikzstyle{small map}=[draw,shape=rectangle,minimum height=4mm,minimum width=4mm,fill=white]

\tikzstyle{cnot}=[fill=white,shape=circle,inner sep=-1.4pt]

\tikzstyle{asym hadamard}=[fill=white,draw,shape=NEbox,inner sep=0.6mm,font=\footnotesize,minimum height=4mm]
\tikzstyle{asym hadamard conj}=[fill=white,draw,shape=NWbox,inner sep=0.6mm,font=\footnotesize,minimum height=4mm]
\tikzstyle{asym hadamard dag}=[fill=white,draw,shape=SEbox,inner sep=0.6mm,font=\footnotesize,minimum height=4mm]

\tikzstyle{hadamard}=[fill=white,draw,inner sep=0.6mm,font=\footnotesize,minimum height=4mm,minimum width=4mm]
\tikzstyle{small hadamard}=[fill=white,draw,inner sep=0.6mm,minimum height=1.5mm,minimum width=1.5mm]
\tikzstyle{small hadamard rotate}=[small hadamard,rotate=45]
\tikzstyle{dhadamard}=[hadamard,doubled]
\tikzstyle{small dhadamard}=[small hadamard,doubled]
\tikzstyle{small dhadamard rotate}=[small hadamard rotate,doubled]
\tikzstyle{antipode}=[white dot,inner sep=0.3mm,font=\footnotesize]

\tikzstyle{scalar}=[diamond,draw,inner sep=0.5pt,font=\small]
\tikzstyle{dscalar}=[diamond,doubled, draw,inner sep=0.5pt,font=\small]

\tikzstyle{small box}=[rectangle,inline text,fill=white,draw,minimum height=5mm,yshift=-0.5mm,minimum width=5mm,font=\small]
\tikzstyle{small gray box}=[small box,fill=gray!30]
\tikzstyle{medium box}=[rectangle,inline text,fill=white,draw,minimum height=5mm,yshift=-0.5mm,minimum width=10mm,font=\small]
\tikzstyle{square box}=[small box] 
\tikzstyle{medium gray box}=[small box,fill=gray!30]
\tikzstyle{semilarge box}=[rectangle,inline text,fill=white,draw,minimum height=5mm,yshift=-0.5mm,minimum width=12.5mm,font=\small]
\tikzstyle{large box}=[rectangle,inline text,fill=white,draw,minimum height=5mm,yshift=-0.5mm,minimum width=15mm,font=\small]
\tikzstyle{large gray box}=[small box,fill=gray!30]

\tikzstyle{Bayes box}=[rectangle,fill=black,draw, minimum height=3mm, minimum width=3mm]

\tikzstyle{gray square point}=[small box,fill=gray!50]

\tikzstyle{dphase box white}=[dhadamard]
\tikzstyle{dphase box gray}=[dhadamard,fill=gray!50!white]
\tikzstyle{phase box white}=[hadamard]
\tikzstyle{phase box gray}=[hadamard,fill=gray!50!white]

\tikzstyle{point}=[regular polygon,regular polygon sides=3,draw,scale=0.75,inner sep=-0.5pt,minimum width=9mm,fill=white,regular polygon rotate=180]
\tikzstyle{point nosep}=[regular polygon,regular polygon sides=3,draw,scale=0.75,inner sep=-2pt,minimum width=9mm,fill=white,regular polygon rotate=180]
\tikzstyle{copoint}=[regular polygon,regular polygon sides=3,draw,scale=0.75,inner sep=-0.5pt,minimum width=9mm,fill=white]
\tikzstyle{dpoint}=[point,doubled]
\tikzstyle{dcopoint}=[copoint,doubled]

\tikzstyle{pointgrow}=[shape=cornerpoint,kpoint common,scale=0.75,inner sep=3pt]
\tikzstyle{pointgrow dag}=[shape=cornercopoint,kpoint common,scale=0.75,inner sep=3pt]

\tikzstyle{wide copoint}=[fill=white,draw,shape=isosceles triangle,shape border rotate=90,isosceles triangle stretches=true,inner sep=0pt,minimum width=1.5cm,minimum height=6.12mm]
\tikzstyle{wide point}=[fill=white,draw,shape=isosceles triangle,shape border rotate=-90,isosceles triangle stretches=true,inner sep=0pt,minimum width=1.5cm,minimum height=6.12mm,yshift=-0.0mm]
\tikzstyle{wide point plus}=[fill=white,draw,shape=isosceles triangle,shape border rotate=-90,isosceles triangle stretches=true,inner sep=0pt,minimum width=1.74cm,minimum height=7mm,yshift=-0.0mm]

\tikzstyle{wide dpoint}=[fill=white,doubled,draw,shape=isosceles triangle,shape border rotate=-90,isosceles triangle stretches=true,inner sep=0pt,minimum width=1.5cm,minimum height=6.12mm,yshift=-0.0mm]

\tikzstyle{tinypoint}=[regular polygon,regular polygon sides=3,draw,scale=0.55,inner sep=-0.15pt,minimum width=6mm,fill=white,regular polygon rotate=180]

\tikzstyle{white point}=[point]
\tikzstyle{white dpoint}=[dpoint]
\tikzstyle{green point}=[white point] 
\tikzstyle{white copoint}=[copoint]
\tikzstyle{gray point}=[point,fill=gray!40!white]
\tikzstyle{gray dpoint}=[gray point,doubled]
\tikzstyle{red point}=[gray point] 
\tikzstyle{gray copoint}=[copoint,fill=gray!40!white]
\tikzstyle{gray dcopoint}=[gray copoint,doubled]

\tikzstyle{white point guide}=[regular polygon,regular polygon sides=3,font=\scriptsize,draw,scale=0.65,inner sep=-0.5pt,minimum width=9mm,fill=white,regular polygon rotate=180]

\tikzstyle{black point}=[point,fill=black,font=\color{white}]
\tikzstyle{black copoint}=[copoint,fill=black,font=\color{white}]

\tikzstyle{tiny gray point}=[tinypoint,fill=gray!40!white]

\tikzstyle{diredge}=[->]
\tikzstyle{ddiredge}=[<->]
\tikzstyle{rdiredge}=[<-]
\tikzstyle{thickdiredge}=[->, very thick]
\tikzstyle{pointer edge}=[->,very thick,gray]
\tikzstyle{pointer edge part}=[very thick,gray]
\tikzstyle{dashed edge}=[dashed]
\tikzstyle{thick dashed edge}=[very thick,dashed]
\tikzstyle{thick gray dashed edge}=[thick dashed edge,gray!40]
\tikzstyle{thick map edge}=[very thick,|->]

\makeatletter
\newcommand{\boxshape}[3]{%
\pgfdeclareshape{#1}{
\inheritsavedanchors[from=rectangle] 
\inheritanchorborder[from=rectangle]
\inheritanchor[from=rectangle]{center}
\inheritanchor[from=rectangle]{north}
\inheritanchor[from=rectangle]{south}
\inheritanchor[from=rectangle]{west}
\inheritanchor[from=rectangle]{east}
\backgroundpath{
\southwest \pgf@xa=\pgf@x \pgf@ya=\pgf@y
\northeast \pgf@xb=\pgf@x \pgf@yb=\pgf@y

\@tempdima=#2
\@tempdimb=#3

\pgfpathmoveto{\pgfpoint{\pgf@xa - 5pt + \@tempdima}{\pgf@ya}}
\pgfpathlineto{\pgfpoint{\pgf@xa - 5pt - \@tempdima}{\pgf@yb}}
\pgfpathlineto{\pgfpoint{\pgf@xb + 5pt + \@tempdimb}{\pgf@yb}}
\pgfpathlineto{\pgfpoint{\pgf@xb + 5pt - \@tempdimb}{\pgf@ya}}
\pgfpathlineto{\pgfpoint{\pgf@xa - 5pt + \@tempdima}{\pgf@ya}}
\pgfpathclose
}
}}

\boxshape{NEbox}{0pt}{3pt}
\boxshape{SEbox}{0pt}{-3pt}
\boxshape{NWbox}{5pt}{0pt}
\boxshape{SWbox}{-5pt}{0pt}
\boxshape{EBox}{-3pt}{3pt}
\boxshape{WBox}{3pt}{-3pt}
\makeatother

\tikzstyle{cloud}=[shape=cloud,draw,minimum width=1.5cm,minimum height=1.5cm]

\tikzstyle{map}=[draw,shape=NEbox,inner sep=1pt,minimum height=4mm,fill=white]
\tikzstyle{dashedmap}=[draw,dashed,shape=NEbox,inner sep=2pt,minimum height=6mm,fill=white]
\tikzstyle{mapdag}=[draw,shape=SEbox,inner sep=1pt,minimum height=4mm,fill=white]
\tikzstyle{mapadj}=[draw,shape=SEbox,inner sep=2pt,minimum height=6mm,fill=white]
\tikzstyle{maptrans}=[draw,shape=SWbox,inner sep=2pt,minimum height=6mm,fill=white]
\tikzstyle{mapconj}=[draw,shape=NWbox,inner sep=2pt,minimum height=6mm,fill=white]

\tikzstyle{medium map}=[draw,shape=NEbox,inner sep=2pt,minimum height=6mm,fill=white,minimum width=7mm]
\tikzstyle{medium map dag}=[draw,shape=SEbox,inner sep=2pt,minimum height=6mm,fill=white,minimum width=7mm]
\tikzstyle{medium map adj}=[draw,shape=SEbox,inner sep=2pt,minimum height=6mm,fill=white,minimum width=7mm]
\tikzstyle{medium map trans}=[draw,shape=SWbox,inner sep=2pt,minimum height=6mm,fill=white,minimum width=7mm]
\tikzstyle{medium map conj}=[draw,shape=NWbox,inner sep=2pt,minimum height=6mm,fill=white,minimum width=7mm]
\tikzstyle{semilarge map}=[draw,shape=NEbox,inner sep=2pt,minimum height=6mm,fill=white,minimum width=9.5mm]
\tikzstyle{semilarge map trans}=[draw,shape=SWbox,inner sep=2pt,minimum height=6mm,fill=white,minimum width=9.5mm]
\tikzstyle{semilarge map adj}=[draw,shape=SEbox,inner sep=2pt,minimum height=6mm,fill=white,minimum width=9.5mm]
\tikzstyle{semilarge map dag}=[draw,shape=SEbox,inner sep=2pt,minimum height=6mm,fill=white,minimum width=9.5mm]
\tikzstyle{semilarge map conj}=[draw,shape=NWbox,inner sep=2pt,minimum height=6mm,fill=white,minimum width=9.5mm]
\tikzstyle{large map}=[draw,shape=NEbox,inner sep=2pt,minimum height=6mm,fill=white,minimum width=12mm]
\tikzstyle{large map conj}=[draw,shape=NWbox,inner sep=2pt,minimum height=6mm,fill=white,minimum width=12mm]
\tikzstyle{very large map}=[draw,shape=NEbox,inner sep=2pt,minimum height=6mm,fill=white,minimum width=17mm]

\tikzstyle{medium dmap}=[draw,doubled,shape=NEbox,inner sep=2pt,minimum height=6mm,fill=white,minimum width=7mm]
\tikzstyle{medium dmap dag}=[draw,doubled,shape=SEbox,inner sep=2pt,minimum height=6mm,fill=white,minimum width=7mm]
\tikzstyle{medium dmap adj}=[draw,doubled,shape=SEbox,inner sep=2pt,minimum height=6mm,fill=white,minimum width=7mm]
\tikzstyle{medium dmap trans}=[draw,doubled,shape=SWbox,inner sep=2pt,minimum height=6mm,fill=white,minimum width=7mm]
\tikzstyle{medium dmap conj}=[draw,doubled,shape=NWbox,inner sep=2pt,minimum height=6mm,fill=white,minimum width=7mm]
\tikzstyle{semilarge dmap}=[draw,doubled,shape=NEbox,inner sep=2pt,minimum height=6mm,fill=white,minimum width=9.5mm]
\tikzstyle{semilarge dmap trans}=[draw,doubled,shape=SWbox,inner sep=2pt,minimum height=6mm,fill=white,minimum width=9.5mm]
\tikzstyle{semilarge dmap adj}=[draw,doubled,shape=SEbox,inner sep=2pt,minimum height=6mm,fill=white,minimum width=9.5mm]
\tikzstyle{semilarge dmap dag}=[draw,doubled,shape=SEbox,inner sep=2pt,minimum height=6mm,fill=white,minimum width=9.5mm]
\tikzstyle{semilarge dmap conj}=[draw,doubled,shape=NWbox,inner sep=2pt,minimum height=6mm,fill=white,minimum width=9.5mm]
\tikzstyle{large dmap}=[draw,doubled,shape=NEbox,inner sep=2pt,minimum height=6mm,fill=white,minimum width=12mm]
\tikzstyle{large dmap conj}=[draw,doubled,shape=NWbox,inner sep=2pt,minimum height=6mm,fill=white,minimum width=12mm]
\tikzstyle{large dmap trans}=[draw,doubled,shape=SWbox,inner sep=2pt,minimum height=6mm,fill=white,minimum width=12mm]
\tikzstyle{large dmap adj}=[draw,doubled,shape=SEbox,inner sep=2pt,minimum height=6mm,fill=white,minimum width=12mm]
\tikzstyle{large dmap dag}=[draw,doubled,shape=SEbox,inner sep=2pt,minimum height=6mm,fill=white,minimum width=12mm]
\tikzstyle{very large dmap}=[draw,doubled,shape=NEbox,inner sep=2pt,minimum height=6mm,fill=white,minimum width=19.5mm]

\tikzstyle{muxbox}=[draw,shape=rectangle,minimum height=3mm,minimum width=3mm,fill=white]
\tikzstyle{dmuxbox}=[muxbox,doubled]

\tikzstyle{box}=[draw,shape=rectangle,inner sep=2pt,minimum height=6mm,minimum width=6mm,fill=white]
\tikzstyle{dbox}=[draw,doubled,shape=rectangle,inner sep=2pt,minimum height=6mm,minimum width=6mm,fill=white]
\tikzstyle{dmap}=[draw,doubled,shape=NEbox,inner sep=2pt,minimum height=6mm,fill=white]
\tikzstyle{dmapdag}=[draw,doubled,shape=SEbox,inner sep=2pt,minimum height=6mm,fill=white]
\tikzstyle{dmapadj}=[draw,doubled,shape=SEbox,inner sep=2pt,minimum height=6mm,fill=white]
\tikzstyle{dmaptrans}=[draw,doubled,shape=SWbox,inner sep=2pt,minimum height=6mm,fill=white]
\tikzstyle{dmapconj}=[draw,doubled,shape=NWbox,inner sep=2pt,minimum height=6mm,fill=white]

\tikzstyle{ddmap}=[draw,doubled,dashed,shape=NEbox,inner sep=2pt,minimum height=6mm,fill=white]
\tikzstyle{ddmapdag}=[draw,doubled,dashed,shape=SEbox,inner sep=2pt,minimum height=6mm,fill=white]
\tikzstyle{ddmapadj}=[draw,doubled,dashed,shape=SEbox,inner sep=2pt,minimum height=6mm,fill=white]
\tikzstyle{ddmaptrans}=[draw,doubled,dashed,shape=SWbox,inner sep=2pt,minimum height=6mm,fill=white]
\tikzstyle{ddmapconj}=[draw,doubled,dashed,shape=NWbox,inner sep=2pt,minimum height=6mm,fill=white]

\boxshape{sNEbox}{0pt}{3pt}
\boxshape{sSEbox}{0pt}{-3pt}
\boxshape{sNWbox}{3pt}{0pt}
\boxshape{sSWbox}{-3pt}{0pt}
\tikzstyle{smap}=[draw,shape=sNEbox,fill=white]
\tikzstyle{smapdag}=[draw,shape=sSEbox,fill=white]
\tikzstyle{smapadj}=[draw,shape=sSEbox,fill=white]
\tikzstyle{smaptrans}=[draw,shape=sSWbox,fill=white]
\tikzstyle{smapconj}=[draw,shape=sNWbox,fill=white]

\tikzstyle{dsmap}=[draw,dashed,shape=sNEbox,fill=white]
\tikzstyle{dsmapdag}=[draw,dashed,shape=sSEbox,fill=white]
\tikzstyle{dsmaptrans}=[draw,dashed,shape=sSWbox,fill=white]
\tikzstyle{dsmapconj}=[draw,dashed,shape=sNWbox,fill=white]

\boxshape{mNEbox}{0pt}{10pt}
\boxshape{mSEbox}{0pt}{-10pt}
\boxshape{mNWbox}{10pt}{0pt}
\boxshape{mSWbox}{-10pt}{0pt}
\tikzstyle{mmap}=[draw,shape=mNEbox]
\tikzstyle{mmapdag}=[draw,shape=mSEbox]
\tikzstyle{mmaptrans}=[draw,shape=mSWbox]
\tikzstyle{mmapconj}=[draw,shape=mNWbox]

\tikzstyle{mmapgray}=[draw,fill=gray!40!white,shape=mNEbox]
\tikzstyle{smapgray}=[draw,fill=gray!40!white,shape=sNEbox]

\makeatletter

\pgfdeclareshape{cornerpoint}{
\inheritsavedanchors[from=rectangle] 
\inheritanchorborder[from=rectangle]
\inheritanchor[from=rectangle]{center}
\inheritanchor[from=rectangle]{north}
\inheritanchor[from=rectangle]{south}
\inheritanchor[from=rectangle]{west}
\inheritanchor[from=rectangle]{east}
\backgroundpath{
\southwest \pgf@xa=\pgf@x \pgf@ya=\pgf@y
\northeast \pgf@xb=\pgf@x \pgf@yb=\pgf@y

\pgfmathsetmacro{\pgf@shorten@left}{\pgfkeysvalueof{/tikz/shorten left}}
\pgfmathsetmacro{\pgf@shorten@right}{\pgfkeysvalueof{/tikz/shorten right}}

\pgfpathmoveto{\pgfpoint{0.5 * (\pgf@xa + \pgf@xb)}{\pgf@ya - 5pt}}
\pgfpathlineto{\pgfpoint{\pgf@xa - 8pt + \pgf@shorten@left}{\pgf@yb - 1.5 * \pgf@shorten@left}}
\pgfpathlineto{\pgfpoint{\pgf@xa - 8pt + \pgf@shorten@left}{\pgf@yb}}
\pgfpathlineto{\pgfpoint{\pgf@xb + 8pt - \pgf@shorten@right}{\pgf@yb}}
\pgfpathlineto{\pgfpoint{\pgf@xb + 8pt - \pgf@shorten@right}{\pgf@yb - 1.5 * \pgf@shorten@right}}
\pgfpathclose
}
}

\pgfdeclareshape{cornercopoint}{
\inheritsavedanchors[from=rectangle] 
\inheritanchorborder[from=rectangle]
\inheritanchor[from=rectangle]{center}
\inheritanchor[from=rectangle]{north}
\inheritanchor[from=rectangle]{south}
\inheritanchor[from=rectangle]{west}
\inheritanchor[from=rectangle]{east}
\backgroundpath{
\southwest \pgf@xa=\pgf@x \pgf@ya=\pgf@y
\northeast \pgf@xb=\pgf@x \pgf@yb=\pgf@y

\pgfmathsetmacro{\pgf@shorten@left}{\pgfkeysvalueof{/tikz/shorten left}}
\pgfmathsetmacro{\pgf@shorten@right}{\pgfkeysvalueof{/tikz/shorten right}}

\pgfpathmoveto{\pgfpoint{0.5 * (\pgf@xa + \pgf@xb)}{\pgf@yb + 5pt}}
\pgfpathlineto{\pgfpoint{\pgf@xa - 8pt + \pgf@shorten@left}{\pgf@ya + 1.5 * \pgf@shorten@left}}
\pgfpathlineto{\pgfpoint{\pgf@xa - 8pt + \pgf@shorten@left}{\pgf@ya}}
\pgfpathlineto{\pgfpoint{\pgf@xb + 8pt - \pgf@shorten@right}{\pgf@ya}}
\pgfpathlineto{\pgfpoint{\pgf@xb + 8pt - \pgf@shorten@right}{\pgf@ya + 1.5 * \pgf@shorten@right}}
\pgfpathclose
}
}

\makeatother

\pgfkeyssetvalue{/tikz/shorten left}{0pt}
\pgfkeyssetvalue{/tikz/shorten right}{0pt}

\tikzstyle{kpoint common}=[draw,fill=white,inner sep=1pt,minimum height=4mm]
\tikzstyle{kpoint sc}=[shape=cornerpoint,kpoint common]
\tikzstyle{kpoint adjoint sc}=[shape=cornercopoint,kpoint common]
\tikzstyle{kpoint}=[shape=cornerpoint,shorten left=5pt,kpoint common]
\tikzstyle{kpoint adjoint}=[shape=cornercopoint,shorten left=5pt,kpoint common]
\tikzstyle{kpoint conjugate}=[shape=cornerpoint,shorten right=5pt,kpoint common]
\tikzstyle{kpoint transpose}=[shape=cornercopoint,shorten right=5pt,kpoint common]
\tikzstyle{kpoint symm}=[shape=cornerpoint,shorten left=5pt,shorten right=5pt,kpoint common]

\tikzstyle{wide kpoint sc}=[shape=cornerpoint,kpoint common, minimum width=1 cm]
\tikzstyle{wide kpointdag sc}=[shape=cornercopoint,kpoint common, minimum width=1 cm]

\tikzstyle{black kpoint}=[shape=cornerpoint,shorten left=5pt,kpoint common,fill=black,font=\color{white}]

\tikzstyle{black kpoint sm}=[shape=cornerpoint,shorten left=5pt,kpoint common,fill=black,font=\color{white},scale=0.75]

\tikzstyle{black kpoint adjoint}=[shape=cornercopoint,shorten left=5pt,kpoint common,fill=black,font=\color{white}]
\tikzstyle{black kpointadj}=[shape=cornercopoint,shorten left=5pt,kpoint common,fill=black,font=\color{white}]

\tikzstyle{black kpointadj sm}=[shape=cornercopoint,shorten left=5pt,kpoint common,fill=black,font=\color{white},scale=0.75]

\tikzstyle{black dkpoint}=[shape=cornerpoint,shorten left=5pt,kpoint common,fill=black, doubled,font=\color{white}]
\tikzstyle{black dkpoint adjoint}=[shape=cornercopoint,shorten left=5pt,kpoint common,fill=black, doubled,font=\color{white}]
\tikzstyle{black dkpointadj}=[shape=cornercopoint,shorten left=5pt,kpoint common,fill=black, doubled,font=\color{white}]

\tikzstyle{black dkpoint sm}=[shape=cornerpoint,shorten left=5pt,kpoint common,fill=black, doubled,font=\color{white},scale=0.75]
\tikzstyle{black dkpointadj sm}=[shape=cornercopoint,shorten left=5pt,kpoint common,fill=black, doubled,font=\color{white},scale=0.75]

\tikzstyle{kpointdag}=[kpoint adjoint]
\tikzstyle{kpointadj}=[kpoint adjoint]
\tikzstyle{kpointconj}=[kpoint conjugate]
\tikzstyle{kpointtrans}=[kpoint transpose]

\tikzstyle{big kpoint}=[kpoint, minimum width=1.2 cm, minimum height=8mm, inner sep=4pt, text depth=3mm]

\tikzstyle{wide kpoint}=[kpoint, minimum width=1 cm, inner sep=2pt]
\tikzstyle{wide kpointdag}=[kpointdag, minimum width=1 cm, inner sep=2pt]
\tikzstyle{wide kpointconj}=[kpointconj, minimum width=1 cm, inner sep=2pt]
\tikzstyle{wide kpointtrans}=[kpointtrans, minimum width=1 cm, inner sep=2pt]

\tikzstyle{wider kpoint}=[kpoint, minimum width=1.25 cm, inner sep=2pt]
\tikzstyle{wider kpointdag}=[kpointdag, minimum width=1.25 cm, inner sep=2pt]
\tikzstyle{wider kpointconj}=[kpointconj, minimum width=1.25 cm, inner sep=2pt]
\tikzstyle{wider kpointtrans}=[kpointtrans, minimum width=1.25 cm, inner sep=2pt]

\tikzstyle{gray kpoint}=[kpoint,fill=gray!50!white]
\tikzstyle{gray kpointdag}=[kpointdag,fill=gray!50!white]
\tikzstyle{gray kpointadj}=[kpointadj,fill=gray!50!white]
\tikzstyle{gray kpointconj}=[kpointconj,fill=gray!50!white]
\tikzstyle{gray kpointtrans}=[kpointtrans,fill=gray!50!white]

\tikzstyle{gray dkpoint}=[kpoint,fill=gray!50!white,doubled]
\tikzstyle{gray dkpointdag}=[kpointdag,fill=gray!50!white,doubled]
\tikzstyle{gray dkpointadj}=[kpointadj,fill=gray!50!white,doubled]
\tikzstyle{gray dkpointconj}=[kpointconj,fill=gray!50!white,doubled]
\tikzstyle{gray dkpointtrans}=[kpointtrans,fill=gray!50!white,doubled]

\tikzstyle{white label}=[draw,fill=white,rectangle,inner sep=0.7 mm]
\tikzstyle{gray label}=[draw,fill=gray!50!white,rectangle,inner sep=0.7 mm]
\tikzstyle{black label}=[draw,fill=black,rectangle,inner sep=0.7 mm]

\tikzstyle{dkpoint}=[kpoint,doubled]
\tikzstyle{wide dkpoint}=[wide kpoint,doubled]
\tikzstyle{dkpointdag}=[kpoint adjoint,doubled]
\tikzstyle{wide dkpointdag}=[wide kpointdag,doubled]
\tikzstyle{dkcopoint}=[kpoint adjoint,doubled]
\tikzstyle{dkpointadj}=[kpoint adjoint,doubled]
\tikzstyle{dkpointconj}=[kpoint conjugate,doubled]
\tikzstyle{dkpointtrans}=[kpoint transpose,doubled]

\tikzstyle{kscalar}=[kpoint common, shape=EBox, inner xsep=-1pt, inner ysep=3pt,font=\small]
\tikzstyle{kscalarconj}=[kpoint common, shape=WBox, inner xsep=-1pt, inner ysep=3pt,font=\small]

\tikzstyle{spekpoint}=[kpoint sc,minimum height=5mm,inner sep=3pt]
\tikzstyle{spekcopoint}=[kpoint adjoint sc,minimum height=5mm,inner sep=3pt]

\tikzstyle{dspekpoint}=[spekpoint,doubled]
\tikzstyle{dspekcopoint}=[spekcopoint,doubled]

 \tikzstyle{upground}=[circuit ee IEC,thick,ground,rotate=90,scale=2.5]
 \tikzstyle{downground}=[circuit ee IEC,thick,ground,rotate=-90,scale=2.5]
 \tikzstyle{bigground}=[regular polygon,regular polygon sides=3,draw=gray,scale=0.50,inner sep=-0.5pt,minimum width=10mm,fill=gray]

\tikzstyle{arrs}=[-latex,font=\small,auto]
\tikzstyle{arrow plain}=[arrs]
\tikzstyle{arrow dashed}=[dashed,arrs]
\tikzstyle{arrow bold}=[very thick,arrs]
\tikzstyle{arrow hide}=[draw=white!0,-]
\tikzstyle{arrow reverse}=[latex-]
\tikzstyle{cdnode}=[]

%% file: Diagrams/discard.tikz
\begin{tikzpicture}
	\begin{pgfonlayer}{nodelayer}
		\node [style=none] (31) at (-1, 0.25) {};
		\node [style=none] (40) at (-1, -0.25) {};
		\node [style=upground] (46) at (-1, 0.5) {};
	\end{pgfonlayer}
	\begin{pgfonlayer}{edgelayer}
		\draw (40.center) to (31.center);
	\end{pgfonlayer}
\end{tikzpicture}

%% file: Diagrams/composite1.tikz
\begin{tikzpicture}
	\begin{pgfonlayer}{nodelayer}
		\node [style=none] (31) at (0, 0.75) {};
		\node [style=none] (38) at (0, -0.75) {};
		\node [style=none] (39) at (0, 0.75) {};
		\node [style=label] (44) at (1, -0.5) {$A\otimes B$};
	\end{pgfonlayer}
	\begin{pgfonlayer}{edgelayer}
		\draw (38.center) to (39.center);
	\end{pgfonlayer}
\end{tikzpicture}

%% file: Diagrams/deteffect1.tikz
\begin{tikzpicture}
	\begin{pgfonlayer}{nodelayer}
		\node [style=none] (31) at (0, 0.25) {};
		\node [style=none] (40) at (0, -0.75) {};
		\node [style=label] (45) at (1, -0.5) {$A\otimes B$};
		\node [style=upground] (46) at (0, 0.5) {};
	\end{pgfonlayer}
	\begin{pgfonlayer}{edgelayer}
		\draw (40.center) to (31.center);
	\end{pgfonlayer}
\end{tikzpicture}

%% file: Diagrams/causal2.tikz
\begin{tikzpicture}
	\begin{pgfonlayer}{nodelayer}
		\node [style=none] (2) at (0, 0.25) {};
		\node [style=none] (3) at (0, -0.75) {};
		\node [style=upground] (7) at (0, 0.5) {};
	\end{pgfonlayer}
	\begin{pgfonlayer}{edgelayer}
		\draw[qWire] (3.center) to (2.center);
	\end{pgfonlayer}
\end{tikzpicture}

%% file: Diagrams/fidInvId1New.tikz
\begin{tikzpicture}
	\begin{pgfonlayer}{nodelayer}
		\node [style=none] (1) at (-1, 0) {};
		\node [style=none] (7) at (1, 0) {};
		\node [style=up label] (9) at (0, 0) {$\Lambda$};
	\end{pgfonlayer}
	\begin{pgfonlayer}{edgelayer}
		\draw [cWire] (1.center) to (7.center);
	\end{pgfonlayer}
\end{tikzpicture}

%% file: Diagrams/fidPrepNorm1.tikz
\begin{tikzpicture}
	\begin{pgfonlayer}{nodelayer}
		\node [style=none] (0) at (0, -0.25) {};
		\node [style=none] (1) at (-1, -0.25) {};
		\node [style=up label] (4) at (-0.75, -0.25) {$\Lambda_S$};
		\node [style=rightground] (5) at (0.25, -0.25) {};
	\end{pgfonlayer}
	\begin{pgfonlayer}{edgelayer}
		\draw [cWire] (1.center) to (0.center);
	\end{pgfonlayer}
\end{tikzpicture}

%% file: Diagrams/fidMeasNorm1.tikz
\begin{tikzpicture}
	\begin{pgfonlayer}{nodelayer}
		\node [style=none] (0) at (0, 0.5) {};
		\node [style=none] (2) at (0, -0.5) {};
		\node [style=right label] (3) at (0, -0.25) {$S$};
		\node [style=upground] (5) at (0, 0.75) {};
	\end{pgfonlayer}
	\begin{pgfonlayer}{edgelayer}
		\draw [qWire] (0.center) to (2.center);
	\end{pgfonlayer}
\end{tikzpicture}

%% file: Diagrams/fidTransNorm1.tikz
\begin{tikzpicture}
	\begin{pgfonlayer}{nodelayer}
		\node [style=none] (5) at (0, 0) {};
		\node [style=none] (6) at (-1, 0) {};
		\node [style=up label] (9) at (-0.75, 0) {$\Lambda_S$};
		\node [style=rightground] (10) at (0.25, 0) {};
	\end{pgfonlayer}
	\begin{pgfonlayer}{edgelayer}
		\draw [cWire] (6.center) to (5.center);
	\end{pgfonlayer}
\end{tikzpicture}

%% file: Diagrams/fidInvId1.tikz
\begin{tikzpicture}
	\begin{pgfonlayer}{nodelayer}
		\node [style=none] (1) at (-1, 0) {};
		\node [style=none] (7) at (1, 0) {};
		\node [style=up label] (9) at (0, 0) {$\Lambda_S$};
	\end{pgfonlayer}
	\begin{pgfonlayer}{edgelayer}
		\draw [cWire] (1.center) to (7.center);
	\end{pgfonlayer}
\end{tikzpicture}

%% file: Diagrams/fidInvId3.tikz
\begin{tikzpicture}
	\begin{pgfonlayer}{nodelayer}
		\node [style=right label] (3) at (0, 0) {$S$};
		\node [style=none] (10) at (0, 0.75) {};
		\node [style=none] (14) at (0, -0.75) {};
	\end{pgfonlayer}
	\begin{pgfonlayer}{edgelayer}
		\draw [qWire] (10.center) to (14.center);
	\end{pgfonlayer}
\end{tikzpicture}

%% file: Diagrams/IdDecomp1.tikz
\begin{tikzpicture}
	\begin{pgfonlayer}{nodelayer}
		\node [style=right label] (13) at (0, 0) {$S$};
		\node [style=none] (14) at (0, -0.75) {};
		\node [style=none] (15) at (0, 0.75) {};
	\end{pgfonlayer}
	\begin{pgfonlayer}{edgelayer}
		\draw [qWire] (14.center) to (15.center);
	\end{pgfonlayer}
\end{tikzpicture}

%% file: Diagrams/MPDiscard2.tikz
\begin{tikzpicture}
	\begin{pgfonlayer}{nodelayer}
		\node [style=copoint] (38) at (0, 0.25) {$e_i$};
		\node [style=none] (39) at (0, -0.75) {};
		\node [style=right label] (42) at (0, -0.75) {$S$};
	\end{pgfonlayer}
	\begin{pgfonlayer}{edgelayer}
		\draw [qWire] (38) to (39.center);
	\end{pgfonlayer}
\end{tikzpicture}

%% file: Diagrams/MPDiscard3.tikz
\begin{tikzpicture}
	\begin{pgfonlayer}{nodelayer}
		\node [style=none] (38) at (0, 0) {};
		\node [style=none] (39) at (0, -0.75) {};
		\node [style=right label] (42) at (0, -0.75) {$S$};
		\node [style=upground] (43) at (0, 0.25) {};
	\end{pgfonlayer}
	\begin{pgfonlayer}{edgelayer}
		\draw [qWire] (38.center) to (39.center);
	\end{pgfonlayer}
\end{tikzpicture}

%% file: Diagrams/Proof12.tikz
\begin{tikzpicture}
	\begin{pgfonlayer}{nodelayer}
		\node [style=none] (31) at (-1.5, 0) {};
		\node [style=fMeas] (38) at (-2.25, 0) {};
		\node [style=none] (39) at (-2.25, -1.25) {};
		\node [style=rightground] (42) at (-1.25, 0) {};
	\end{pgfonlayer}
	\begin{pgfonlayer}{edgelayer}
		\draw [cWire] (38) to (31.center);
		\draw [qWire] (38) to (39.center);
	\end{pgfonlayer}
\end{tikzpicture}

%% file: Diagrams/Proof13.tikz
\begin{tikzpicture}
	\begin{pgfonlayer}{nodelayer}
		\node [style=none] (38) at (-0.5, 0.25) {};
		\node [style=none] (39) at (-0.5, -0.75) {};
		\node [style=upground] (42) at (-0.5, 0.5) {};
	\end{pgfonlayer}
	\begin{pgfonlayer}{edgelayer}
		\draw [qWire] (38.center) to (39.center);
	\end{pgfonlayer}
\end{tikzpicture}

%% file: Diagrams/newQuasi3.tikz
\begin{tikzpicture}
	\begin{pgfonlayer}{nodelayer}
		\node [style=infpoint, fill=black] (78) at (-1.25, 0) {\color{white}$Q'$};
		\node [style=none] (79) at (-1, 0) {};
		\node [style=infcopoint] (80) at (0, 0) {$\alpha$};
	\end{pgfonlayer}
	\begin{pgfonlayer}{edgelayer}
		\draw [cWire] (79.center) to (80);
	\end{pgfonlayer}
\end{tikzpicture}

%% file: Diagrams/mu.tikz
\begin{tikzpicture}
	\begin{pgfonlayer}{nodelayer}
		\node [style=point] (0) at (0, 0) {$\mu$};
		\node [style=none] (3) at (0, 0.75) {};
		\node [style=right label] (5) at (0, 0.5) {$T$};
	\end{pgfonlayer}
	\begin{pgfonlayer}{edgelayer}
		\draw [qWire] (0) to (3.center);
	\end{pgfonlayer}
\end{tikzpicture}

%% file: Diagrams/s.tikz
\begin{tikzpicture}
	\begin{pgfonlayer}{nodelayer}
		\node [style=point] (0) at (0, 0) {$s$};
		\node [style=none] (3) at (0, 0.75) {};
		\node [style=right label] (5) at (0, 0.5) {$T$};
	\end{pgfonlayer}
	\begin{pgfonlayer}{edgelayer}
		\draw [qWire] (0) to (3.center);
	\end{pgfonlayer}
\end{tikzpicture}

%% file: Diagrams/trace.tikz
\begin{tikzpicture}
	\begin{pgfonlayer}{nodelayer}
		\node [style=none] (38) at (0, -0.25) {};
		\node [style=none] (39) at (0, -0.75) {};
		\node [style=right label] (42) at (0, -0.75) {$S$};
		\node [style=upground] (44) at (0, 0) {};
	\end{pgfonlayer}
	\begin{pgfonlayer}{edgelayer}
		\draw [qWire] (38.center) to (39.center);
	\end{pgfonlayer}
\end{tikzpicture}

%% file: Diagrams/e.tikz
\begin{tikzpicture}
	\begin{pgfonlayer}{nodelayer}
		\node [style=copoint] (1) at (0, 0) {$e$};
		\node [style=none] (2) at (0, -0.75) {};
		\node [style=right label] (4) at (0, -0.75) {$S$};
	\end{pgfonlayer}
	\begin{pgfonlayer}{edgelayer}
		\draw [qWire] (1) to (2.center);
	\end{pgfonlayer}
\end{tikzpicture}